\documentclass{amsart}

\usepackage[top=1in, bottom=1in, left=1in, right=1in]{geometry}
\usepackage{lineno}
\usepackage[english]{babel}
\usepackage{framed}
\usepackage[normalem]{ulem}
\usepackage{mathtools}
\usepackage{amsmath}
\usepackage{amsthm}
\usepackage{amssymb}
\usepackage{amsfonts}
\usepackage{a4wide}
\usepackage{booktabs,tabularx,array}
\usepackage{enumerate}
\usepackage{appendix}
\usepackage{stmaryrd}
\usepackage{hyperref}
\usepackage{url}
\usepackage{latexsym}
\usepackage{indentfirst}
\usepackage{placeins}
\usepackage{booktabs}
\usepackage{algorithm}
\usepackage[noend]{algpseudocode}
\usepackage{array, makecell}
\usepackage{color,xcolor,colortbl}
\usepackage{subfig}
\usepackage{graphicx,color} 
\usepackage{diagbox}
\usepackage[capitalize]{cleveref}
\usepackage{braket}
\usepackage{nicematrix}
\usepackage{tikz}
\usetikzlibrary{quantikz2}
\usepackage{mathrsfs}

\theoremstyle{plain}
\newtheorem{theorem}{Theorem}[section]
\newtheorem{lemma}[theorem]{Lemma}
\newtheorem{prop}[theorem]{Proposition}
\newtheorem{cor}[theorem]{Corollary}

\theoremstyle{definition}
\newtheorem{defn}{Definition}[section]

\theoremstyle{remark}
\newtheorem{rmk}[theorem]{Remark}

\DeclareMathOperator{\B}{B}
\DeclareMathOperator{\U}{U}
\DeclareMathOperator{\PU}{PU}

\newcommand{\M}{\operatorname{M}}

\newcommand{\Sp}{\operatorname{Sp}}
\newcommand{\SL}{\operatorname{SL}}
\newcommand{\GL}{\operatorname{GL}}

\newcommand{\Cliff}{\mathcal C}
\newcommand{\SC}{\mathcal{SC}}

\definecolor{gscLine}{HTML}{225EA8}
\definecolor{gscFill}{HTML}{DEEBF7}
\definecolor{piLine}{HTML}{2B8CBE}
\definecolor{piFill}{HTML}{A6DCEF}
\definecolor{scLine}{HTML}{D95F02}
\definecolor{scFill}{HTML}{FDD0A2}
\definecolor{formLine}{HTML}{E68600}
\definecolor{formFill}{HTML}{FEE8B0}
\definecolor{lscLine}{HTML}{6A51A3}
\definecolor{lscFill}{HTML}{DADAEB}

\tikzset{
  set boundary/.style={line width=1.15pt},
  set label/.style={align=center, font=\normalsize},
  main label/.style={align=center, font=\normalsize\bfseries}
}

\begin{document}
\raggedbottom

\title[Semi-Cliffordness of the Clifford hierarchy]{Semi-Cliffordness of the Clifford hierarchy for a single qudit in composite dimensions}

\author{Yifei Qi}
\address{\'Ecole polytechnique f\'ed\'erale de Lausanne (EPFL), Lausanne, Switzerland}
\curraddr{}
\email{yifei.qi@epfl.ch}
\thanks{}

\author{Rahul Sarkar}
\address{Department of Mathematics, University of California, Berkeley, CA 94720}
\curraddr{}
\email{rsarkar@berkeley.edu}
\thanks{}



\date{}

\dedicatory{}

\commby{}
\begin{abstract}
An important question in quantum information theory asks whether every gate in the Clifford hierarchy is semi-Clifford or not, as such gates admit resource-efficient implementations via gate teleportation. Extending recent results for prime dimensions~\cite{silva_clifford_2025}, we prove that every hierarchy gate on a single qudit of dimension \(d\) is semi-Clifford, if and only if \(d\) is square-free. In composite dimensions, \(\mathbb{Z}_d^2\) is a symplectic module rather than a vector space, motivating a distinction between four types of gates. Semi-Clifford gates ($\mathcal{SC}$) are the ones that become diagonal under left and right multiplication by Clifford gates, while a broader class $\mathcal{N}$ admits a permutation diagonal form, under left and right multiplication by Cliffords. Lagrangian semi-Clifford gates ($\mathcal{LSC}$) conjugate some maximal abelian Pauli subgroup to another, whereas generalized semi-Clifford gates ($\mathcal{GSC}$) map some maximal abelian Pauli subalgebra to another under conjugation. In square-free dimensions, we have \(\mathcal{SC}=\mathcal{LSC}\) and \(\mathcal{N}=\mathcal{GSC}\), but when \(d\) is not square-free, these equivalences can fail, as Lagrangian submodules can become non-free. We also show that every third-level gate of the one-qudit Clifford hierarchy is a generalized semi-Clifford in arbitrary dimensions.
\end{abstract}

\maketitle

\section{Introduction}

Quantum error correction is essential for realizing large-scale quantum computation, since quantum information must remain protected while logical operations are performed~\cite{preskill_quantum_2018}. Pauli operators, generated by the shift and phase operators, form an orthogonal basis for the space of linear operators (Definition~\ref{def:pauli_weyl}), so any error operator can be expressed as a linear combination of them. Clifford gates are precisely the unitaries that preserve this structure under conjugation, mapping each Pauli operator to another Pauli operator. This property simplifies error tracking through Clifford circuits and underlies their central role in fault-tolerant quantum computation. However, Clifford gates alone are not universal, and the non-Clifford operations are generally more costly to realize fault-tolerantly~\cite{Gottesman_FT_1998, Eastin_Knill_2009, Bravyi_Kitaev_2005}. A standard approach to implement such non-Clifford gates is quantum gate teleportation~\cite{Gottesman_1999}. In this protocol, a gate \(U\) is implemented using a gate-dependent resource state, and the measurement outcome produces a Pauli byproduct \(P\). Correcting this byproduct requires implementing \(UPU^*\). This recursive reduction motivates the Clifford hierarchy, that is, for \(k\geq 2\),
\(\mathcal{C}_k =\left\{U : UPU^*\in\mathcal{C}_{k-1}\text{ for every Pauli operator } P\right\}\), where $U$ is also unitary (Definition~\ref{def:clifford-hierarchy}).
In particular, the first and second hierarchies are the Pauli and Clifford groups, respectively. Thus, gate teleportation simplifies the implementation of a non-Clifford gate by shifting the difficulty to preparing a gate-dependent resource state.  Within this hierarchy, semi-Clifford gates (Definition~\ref{def:semi-clifford}) are those that can be made diagonal by Clifford pre- and post-processing. This additional structure permits more resource-efficient teleportation schemes. More specifically, whereas standard teleportation uses a Bell measurement and two ancillary qubits per teleported qubit, diagonal gates (Definition~\ref{def:diagonal-gates}) use a one-bit teleportation that reduces the ancillary overhead~\cite{Zhou_2000}.

Although the preceding discussion has focused primarily on qubits, higher-dimensional qudit systems have attracted increasing interest due to their potential advantages in quantum computation, simulation, and physical implementation; see \cite{Wang_2020} for a review. The Pauli--Clifford framework underlying quantum error correction has likewise been extended from qubits to qudits, both in prime dimensions and in arbitrary dimensions~\cite {Gottesman1999FaultTolerant, Hostens_2005}. More recently, stabilizer-code constructions tailored specifically to composite dimensions have highlighted the ring-theoretic and symplectic structures required in this setting~\cite{Gunderman_2025}.

\subsection{Related work}
The studies between Clifford hierarchy and semi-Clifford gates were first done for qubits (\(d=2\)) and subsequently extended to higher-dimensional qudits (\(d>2\)) \cite{silva_efficient_2021}, raising a more general question to determine the triples \((d,n,k)\) for which every hierarchy gate on \(n\) qudits of dimension \(d\) at the \(k\)th level of the Clifford hierarchy is semi-Clifford. 

Recent work on the single-qudit Clifford hierarchy in prime dimension shows that every hierarchy gate is semi-Clifford~\cite{silva_clifford_2025}. In contrast, multi-qubit hierarchy gates need not be semi-Clifford, motivating the introduction of generalized semi-Clifford operators~\cite{zeng_semi-clifford_2008}. In the multi-qubit setting, semi-Cliffordness is defined by the requirement that a unitary conjugate some maximal abelian Pauli subgroup onto another. This is equivalent to the normal form $U=C_1DC_2$, where $C_1,C_2$ are Clifford operators and $D$ is diagonal. Generalized semi-Cliffordness requires only that the complex linear span of some maximal abelian Pauli subgroup be conjugated onto another such span; its corresponding normal form is $U=C_1\Pi DC_2$, where $\Pi$ is a computational-basis permutation~\cite{zeng_semi-clifford_2008}. The span condition preserves the associated maximal abelian operator algebra without requiring individual Pauli operators to remain Pauli operators under conjugation. These characterizations also extend to prime-dimensional qudits~\cite{silva_efficient_2021}. Further structural results on semi-Clifford and generalized semi-Clifford gates within the qubit hierarchy were obtained in \cite{Anderson_2024}, while permutation gates have more recently been studied through affine-equivalence methods~\cite{anderson2025affineequivalence}. For general qudit dimensions, the subgroup and algebra conditions must be distinguished from their respective normal forms, since maximal abelian Pauli subgroups need not all be equivalent under the conjugation by Clifford gates. To avoid ambiguity, we reserve \emph{semi-Clifford} for $U=C_1DC_2$ and call the subgroup-conjugation condition \emph{Lagrangian semi-Cliffordness}. We retain \emph{generalized semi-Clifford} for the span-conjugation condition and refer separately to $U=C_1\Pi DC_2$ as the \emph{permutation diagonal normal form}. Table~\ref{tab:known-boundaries} summarizes the known
boundaries and the results established in this work.
\begin{table}[t]
\centering
\footnotesize
\setlength{\tabcolsep}{4pt}
\renewcommand{\arraystretch}{1.12}
\begin{tabularx}{\textwidth}{
    @{} c c c >{\raggedright\arraybackslash}X l @{}}
\toprule
\(d\) & \(n\) & Level/sector & Conclusion & Reference \\
\midrule

\(2\)
& \(1,2\)
& all \(k\)
& Every hierarchy gate is \(\mathcal{SC}\).
& \cite{zeng_semi-clifford_2008} \\

\(2\)
& \(3\)
& \(k=3\)
& Every hierarchy gate is \(\mathcal{SC}\).
& \cite{zeng_semi-clifford_2008} \\

\(2\)
& \(4\)
& \(k=3\)
& Every hierarchy gate is \(\mathcal{SC}\).
& \cite{anderson2025affineequivalence} \\

\(2\)
& \(n\ge 7\)
& \(k=3\)
& Non-\(\mathcal{SC}\) hierarchy gates exist, beginning with the seven-qubit Gottesman--Mochon counterexample.
& \cite{beigi2009c3semicliffordgeneralizedsemiclifford} \\

\(2\)
& arbitrary
& \(k=3\)
& Every hierarchy gate is \(\mathcal{GSC}\).
& \cite{beigi2009c3semicliffordgeneralizedsemiclifford,Pllaha_2020} \\

\(2\)
& \(n>2\)
& \(k>3\)
& Non-\(\mathcal{SC}\) hierarchy gates exist.
& \cite{zeng_semi-clifford_2008} \\

\(2\)
& \(n=5\)
& \(k=5\)
& The generalized semi-Clifford conjecture is false.
& \cite{desilva2026generalisedsemicliffordconjecturefalse} \\

prime
& arbitrary
& diagonal, all \(k\)
& The diagonal hierarchy is completely classified.
& \cite{Cui_2017} \\

prime
& \(1\); also \(n=2\) for \(p=3\)
& \(k=3\)
& Every hierarchy gate is \(\mathcal{SC}\) and admits efficient semi-Clifford teleportation.
& \cite{silva_efficient_2021} \\

odd prime
& \(n\le 2\)
& \(k=3\)
& Every hierarchy gate is \(\mathcal{SC}\).
& \cite{chen2024characterisingsemicliffordgatesusing} \\

odd prime
& \(1\)
& all \(k\)
& Every hierarchy gate is \(\mathcal{SC}\); each level has a unique normal form and an explicit cardinality.
& \cite{silva_clifford_2025} \\

prime
& arbitrary
& \(k\leq 4\)
& Every hierarchy gate is \(\mathcal{GSC}\).
&\cite{marcus2026generalizedsemicliffordconjectureholds}\\

all \(d\geq2\)
& \(1\)
& all \(k\)
& Every hierarchy gate is semi-Clifford if and only if
  \(d\) is square-free.
& Theorem~\ref{thm:semi-clifford-equivalent} \\

all \(d\geq2\)
& \(1\)
& \(k=3\)
& Every hierarchy gate is generalized semi-Clifford.
& Theorem~\ref{thm:level-three-gsc-arbitrary} \\

\bottomrule
\end{tabularx}
\caption{Known semi-Clifford boundaries. \(\mathcal{SC}\) and \(\mathcal{GSC}\) denote semi-Clifford and generalized semi-Clifford, respectively.}
\label{tab:known-boundaries}
\end{table}

In this work, we consider an arbitrary dimension \(d\), with the underlying Pauli–Clifford structure governed by the symplectic module \(\mathbb Z_d^2\). Its symplectic form encodes Pauli commutation relations, while its Lagrangian geometry underlies the semi-Clifford structure. To handle composite dimensions, we use projective geometry over finite rings, matrix reduction, and finite character duality as ring-theoretic substitutes for the vector-space methods available in prime dimensions~\cite{tignol1986symplectic,havlicek_projective_2008,albouy2008matrix,brown1993matrices,conrad_characters_2007}.

\subsection{Main results}
\label{ssec:main-results}
In Section~\ref{ssec:semi_clifford_generalized}, we first provide equivalent conditions for four types of gates, namely semi-Clifford gates \(\mathcal{SC}\), gates admitting permutation diagonal form \(\mathcal N\), Lagrangian semi-Clifford gates \(\mathcal{LSC}\), and generalized semi-Clifford gates \(\mathcal{GSC}\) (Lemma~\ref{lem:unimodular-direction-criteria} and Lemma~\ref{lem:pauli-conjugation-equivalence}). Furthermore, we summarize these relations for the four types of gates in Figure~\ref{fig:semi-clifford-venn}.

We then establish a sharp dimensional criterion for semi-Cliffordness, that is, every gate in the single-qudit Clifford hierarchy is semi-Clifford if and only if \(d\) is square-free (Theorem~\ref{thm:semi-clifford-equivalent}). The square-free case (Lemma~\ref{lem:squarefree-semiclifford}) follows from a tensor-product factorization of every hierarchy level into its prime-power components (Lemma~\ref{lem:coprime-hierarchy-factorization}), together with the known prime-dimensional results in~\cite{zeng_semi-clifford_2008,silva_clifford_2025}. In square-free dimensions,  we also prove that the subgroup and algebra conjugation conditions are equivalent to the corresponding normal forms, that is, \(\mathcal{LSC}=\mathcal{SC}\) and \(\mathcal{GSC}=\mathcal N\)(Theorem~\ref{thm:square-free-class-equivalence}).  In contrast, when \(d\) is not square-free, non-free Lagrangians can break these equivalences between semi-Clifford and Clifford hierarchy gates, as demonstrated by the dimension-nine examples in Section~\ref{ssec:d9-prototype}. This construction generalizes to a third-level gate that is not semi-Clifford in a single qudit when \(d\) is not square free~\ref{lem:explicit-clifford-shears}. Besides, we derive a more complex relation among the four types of gates when \(d\) is not square-free in Figure~\ref{fig:dim9-relations}.

Finally, we show that every third-level gate on a single qudit is generalized semi-Clifford in arbitrary dimension $d$ (Theorem~\ref{thm:level-three-gsc-arbitrary}). 

\subsection*{Acknowledgments} Y.Q. was supported by The Chinese University of Hong Kong, Shenzhen. She also acknowledges the UC Berkeley Concurrent Enrollment Program, through which this collaboration was initiated. R.S. was supported by the U.S. Department of Energy, Office of Science, Accelerated Research in Quantum Computing Centers, Quantum Utility through Advanced
Computational Quantum Algorithms, grant no. DE-SC0025572. 

\subsection*{On AI use} The formal planning, main results, mathematical constructions, and composition for this document were completed by the authors. During the exploratory stages of this work, AI tools were consulted asynchronously as supplemental aids; for example, the reasoning agent Rethlas \cite{ju2026automated} helped identify candidate structures for the dimension-9 counterexample. All AI outputs were rigorously reviewed, independently verified, and implemented serially.
\section{Preliminaries}
\label{sec:prelim}

Let $\mathbb{N}$ denote the set of positive integers, $\mathbb{Z}$ denote the set of integers, and $\mathbb Q$ denote the set of rational numbers. Fix an integer \(d\ge 2\), which will denote the \textit{qudit dimension}, and write its prime factorization as $d\coloneqq \prod_{j=1}^s p_j^{\alpha_j}$, where \(p_1,\dots,p_s\) are distinct primes, \(\alpha_j\in \mathbb N\) for every \(j\), and \(s\) denotes the number of distinct prime divisors of \(d\). We write \(\mathbb{Z}_d:= \mathbb Z / d \mathbb Z\) and \(\mathbb{Z}_d^\times\) for its group of units. We denote the Hilbert space of the qudit by \(\mathcal H := \mathbb C^d\), with computational basis \(\{\, |x\rangle : x\in \mathbb{Z}_d\,\}\), and set \(\omega := e^{2\pi i/d}\) and \(\tau:= e^{\pi i/d}\). The dimension \(d\) is \textit{square-free} when every \(\alpha_j=1\), and \textit{non-square-free} when \(\alpha_j\ge 2\) for at least one \(j\). For each prime-power factor, we denote \(\mathcal{H}_j:=\mathbb C^{q_j}\) and \(\omega_j:=e^{2\pi i/q_j}\).

We denote the set of linear operators on \(\mathcal H\) as \(\B(\mathcal H) :=\operatorname{End}_{\mathbb C}(\mathcal{H})\). We equip \(\B(\mathcal H)\) with the Hilbert--Schmidt inner product \(\langle A,B\rangle_{\mathrm{HS}}:=\operatorname{Tr}(A^*B)\), where \(\operatorname{Tr}(\cdot)\) denotes the trace and \(A^*\) denotes the adjoint. The identity matrix will be denoted \(I\) of the relevant dimension.  We write \(\U(\mathcal{H})\) for the \textit{unitary group} on $\mathcal{H}$ and use the term \emph{gate} to mean a unitary operator on \(\mathcal{H}\). We denote \(\mathbb T := \{z\in \mathbb C: |z|=1\}\) for the \textit{circle group}, and $\PU(\mathcal{H}) := \U(\mathcal{H}) / \mathbb{T} I$ for the \textit{projective unitary group} as the quotient group. The projective class of a unitary \(U\) is denoted by \([U]\). 

For a set \(S\), we denote its cardinality as \(|S|\). The \emph{order} of a group \(G\) is its cardinality \(|G|\), and the order of an element \(g\in G\) is the cardinality of the cyclic subgroup \(\langle g\rangle\); an element has infinite order if this subgroup is infinite. For a prime \(p\), a finite group is called a \(p\)-group if its order is a power of \(p\). 

\subsection{Symplectic Modules}
\label{ssec:symplectic-modules}

Let \(R\) be a commutative ring with identity. For \(k\in\mathbb N\), we write \(\M_k(R)\) for the ring of \(k\times k\) matrices over \(R\), \(\GL_k(R)\) for its group of invertible matrices, and \(\SL_k(R)\) for the subgroup of matrices with determinant one. Let \(\mathcal R\) be an \(R\)-module. A subset \(\{v_1,\ldots,v_n\} \subseteq \mathcal R\) is called an \(R\)-basis if \(a_1v_1+\cdots+a_n v_n = 0\) implies \(a_1=\cdots=a_n=0\), where \(a_1,\ldots,a_n\in R\), and the subset is called a \emph{generating set} of a submodule \(\mathcal{R}' \subseteq \mathcal{R}\), if any element \(u \in \mathcal{R}'\) can be expressed as \(u=a_1v_1+\cdots+a_n v_n\), for some \(a_1,\ldots,a_n\in R\). We say that \(\mathcal R\) is a
\emph{free module} if it admits an \(R\)-basis \(\{v_1,\ldots,v_n\}\), and in this case we say that \(\mathcal R\) has \textit{rank} \(n\). In particular, \(R^n\) is free of rank \(n\) with its standard basis \(\{e_1,\cdots,e_n\}\). An \(R\)-bilinear map \(f\colon\mathcal R\times\mathcal R\to R\) is called an \emph{alternating form} if \(f(v,v)=0\) for every \(v\in\mathcal R\). We call \(f\) \emph{nondegenerate} if, for every \(v\in\mathcal R\), the condition \(f(v,v')=0\) for all \(v'\in\mathcal R\) implies \(v=0\). A nondegenerate alternating form is called a \emph{symplectic form}, and the pair \((\mathcal R,f)\) is then called a \emph{symplectic module}.

Let us fix a symplectic module \((\mathcal R,f)\). For an \(R\)-submodule \(L\subseteq\mathcal R\), we define its \emph{\(f\)-orthogonal complement} by \(L^\perp:=\{v\in\mathcal R: f(v,\ell)=0, \; \forall \ell\in L\}\), which is an \(R\)-submodule by bilinearity. We call \(L\) \emph{isotropic} if \(L\subseteq L^\perp\), and \emph{Lagrangian} if \(L=L^\perp\). Notice that a Lagrangian submodule \(L\) is maximally isotropic, that is, \(L\) is not properly contained in another isotropic submodule.

We now fix \(k\geq1\), and specialize to free symplectic modules of even rank \((R^{2k},f)\). An element \(v=(c_1,\ldots,c_{2k})^{\mathsf T}\in R^{2k}\) is \emph{unimodular} if its coordinates generate the unit ideal, that is, \(Rc_1+\cdots+Rc_{2k}=R\). Such an element generates a free submodule \(Rv\) with basis \(\{v\}\), because if \(rv=0\) for some \(r \in R\), then by definition there exist \(b_1,\cdots,b_{2k}\in R\) such that \(\sum_{j=i}^{2k} b_jc_j=1\), giving \(r=\sum_{j=1}^{2k} b_j(rc_j)=0\). Moreover, \(Rv\) is isotropic because \(f\) is alternating. Below, Lemma~\ref{lem:symplectic-size} gives a criterion for an isotropic submodule to be Lagrangian and further shows that every free Lagrangian submodule of \(R^{2k}\) has rank \(k\), and each of its basis elements is unimodular, for any basis. We first recall some basic definitions and facts. For a finite abelian group \(A\), let \(\operatorname{Char}(A)\) denote its \emph{character group}, consisting of all group homomorphisms from \(A\) to \(\mathbb T\), with the group operation being point-wise multiplication; thus in \(\operatorname{Char}(A)\), the product of the characters \(\chi_1\) and \(\chi_2\)  is defined as \((\chi_1\cdot\chi_2)(a):=\chi_1(a)\chi_2(a)\), for all \(a\in A\). Since \(A\) is isomorphic to \(\operatorname{Char}(A)\), we have \(|\operatorname{Char}(A)|=|A|\); see \cite[Theorem~3.13]{conrad_characters_2007}. For a subgroup \(B\subseteq A\), we define its \emph{annihilator} as \(B^\circ : =\{\eta\in\operatorname{Char}(A): \eta(b)=1, \, \forall b\in B\}\), which is a subgroup of \(\operatorname{Char}(A)\). Let \(\pi:A\to A/B\) be the quotient map. Each \(\chi\in B^\circ\) is constant on cosets of \(B\), so it determines a unique character \(\bar\chi\in\operatorname{Char}(A/B)\) by \(\bar\chi(a+B)=\chi(a)\), satisfying \(\chi=\bar\chi\circ\pi\). Conversely, composing any character \(\bar\chi\) of \(A/B\) with \(\pi\), gives a character of \(A\) that is trivial on \(B\), that is, \(\bar\chi\circ\pi\in B^\circ\). Both of the constructions are homomorphisms, yielding the canonical group isomorphism \(B^\circ\cong\operatorname{Char}(A/B)\).

Now given that \(R\) is finite, we say that a character on \((R,+)\) is \emph{faithful} if its kernel contains no nonzero ideal of \(R\), where an ideal of \(R\) is an additive subgroup of \(R\) closed under multiplication by every element of \(R\). When the ring is \(\mathbb Z_d\), it always has a faithful character \(\chi(x)=\omega^x\), for \(x\in\mathbb Z_d\). Given the symplectic module \((R^{2k},f)\), such a character induces a homomorphism
\[
\Psi_{f,\chi}\colon (R^{2k},+)\longrightarrow\operatorname{Char}((R^{2k},+)),
\qquad
\Psi_{f,\chi}(v)(w):=\chi\bigl(f(v,w)\bigr).
\]
For every submodule \(L\subseteq R^{2k}\) and
\(v\in R^{2k}\), the image \(f(v,L):=\{f(v,\ell):\ell\in L\}\) is an ideal of \(R\). Using the nondegeneracy of \(f\) and the faithfulness of \(\chi\), we obtain \(\ker\Psi_{f,\chi}=\{0\}\), so \(\Psi_{f,\chi}\) is injective. Since \(|\operatorname{Char}((R^{2k},+))|=|R^{2k}|\), it follows that \(\Psi_{f,\chi}\) is an isomorphism.  Faithfulness of \(\chi\) therefore gives \(\Psi_{f,\chi}(v)\in L^\circ\) (i.e. \(\chi(f(v,\ell))=1\ \text{for all }\ell\in L\)) if and only if \(f(v,\ell)=0, \ \text{for all }\ell\in L\). Its restriction \(\Psi_{f,\chi}|_{L^\perp}\) identifies \(L^\perp\) with \(L^\circ\), giving \(|L^\perp|=|L^\circ|\); compare \cite[Proposition~2.2]{tignol1986symplectic} and \cite[Exercise~7, Section~3]{conrad_characters_2007}. 

\begin{lemma}
\label{lem:symplectic-size}
Let \(R\) be a finite commutative ring with identity admitting a faithful character \(\chi:(R,+)\to\mathbb T\), and let \((R^{2k},f)\) be a symplectic \(R\)-module. Then every submodule \(L\subseteq R^{2k}\) satisfies \(|L|\,|L^\perp|=|R|^{2k}\). If \(L\) is isotropic, then \(|L|\leq|R|^k\), with equality if and only if \(L\) is Lagrangian. Moreover, if \(\{v_1,\ldots,v_k\}\) is an \(R\)-basis for a free Lagrangian \(L\), then each basis element is unimodular in \(R^{2k}\).
\end{lemma}

\begin{proof}
First we suppose that $L \subseteq R^{2k}$ is a submodule. By the preceding discussion, we have \(|L^\perp| =|L^\circ|=|\operatorname{Char}(R^{2k}/L)|  =|R|^{2k}/|L|\). If \(L\) is isotropic, then \(L\subseteq L^\perp\), so \(|L|\leq|R|^{k}\), with equality exactly when \(L=L^\perp\), that is, Lagrangian.

Suppose now that \(L\) is free and Lagrangian, with basis \(\{v_1,\ldots,v_r\}\), i.e. \(L=Rv_1\oplus\cdots\oplus Rv_r\). Then \(|R|^r=|L|=|R|^k\), so \(r=k\). For \(j=1,\dots,k\), consider the \(R\)-linear map \(\varphi_j:R^{2k}\longrightarrow R\) defined by \(\varphi_j(x):=f(x,v_j)\), whose kernel is \((Rv_j)^\perp\). Since \(Rv_j\cong R\) as abelian groups, we have \(|Rv_j| = |R|\), and then the first isomorphism theorem for groups shows that the image of $\varphi_j$ has cardinality \(|R|^{2k}/|(Rv_j)^\perp|\). Combining with the results of the previous paragraph, we get \(|R|^{2k}/|(Rv_j)^\perp|=|Rv_j|=|R|\), and therefore \(\varphi_j\) is surjective. Therefore, there exists $x \in R^{2k}$ such that \(f(x,v_j) = 1 \in R\), which proves unimodularity of \(v_j\).
\end{proof}

Next, still assuming the symplectic module \((R^{2k},f)\), we define \emph{symplectic group} to be a subgroup of invertible \(R\)-linear map \(\operatorname{Sp}_{2k}(R):=\{S\in\operatorname{GL}_{2k}(R):f(Sv,Sw)=f(v,w), \; \forall v,w\in R^{2k}\}\), that is, \(\operatorname{Sp}_{2k}(R)\) consists of invertible matrices that preserve the bilinear form \(f\). Any element of \(\operatorname{Sp}_{2k}(R)\) is called a \emph{symplectic matrix} and defines a group action on \(R^{2k}\). We also denote \(\operatorname{Lag}(R^{2k})\) to be the set of Lagrangian submodules of \(R^{2k}\). The following lemma describes the natural symplectic action on this set.

\begin{lemma}[Symplectic action on Lagrangian submodules]
\label{lem:symplectic-action-lagrangians}
Let \((R^{2k},f)\) be a symplectic module and \(S\in\operatorname{Sp}_{2k}(R)\).  Then for every \(R\)-submodule \(L\subseteq R^{2k}\), we have \((S(L))^\perp=S(L^\perp)\). Consequently, for a Lagrangian submodule \(L\), \(S\cdot L:=S(L)\) defines a left action of \(\operatorname{Sp}_{2k}(R)\) on \(\operatorname{Lag}(R^{2k})\). This action preserves the freeness of the Lagrangian submodule.
\end{lemma}

\begin{proof}
For a \(R\)-submodule \(L\), as \(S\) preserves \(f\), we have \(x\in(S(L))^\perp\) if and only if \(S^{-1}x\in L^\perp\). This proves the identity \((S(L))^\perp=S(L^\perp)\). Furthermore, if \(L\) is Lagrangian, then \((S(L))^\perp=S(L^\perp)=S(L)\) and this shows that \(S\) sends Lagrangian submodules to Lagrangian submodules. For the second part of the lemma, the action laws follow from composition, and the first part. Since, any element of \(\operatorname{GL}_{2k}(R)\) preserves the freeness of a \(R\)-submodule of \(R^{2k}\) upon left multiplication, the last statement also follows.
\end{proof}

We next classify the Lagrangian submodules of \(\mathbb Z_{p^\alpha}^2\) up to symplectic transformations, where \(p\) is prime and \(\alpha\in\mathbb N\), using the \emph{standard symplectic form} \(f(v,w):=v_2w_1-v_1w_2\). for \(v=(v_1,v_2)^\top\) and \(w=(w_1,w_2)^\top\) in \(\mathbb Z_{p^\alpha}^2\). We note that the symplectic group satisfies \(\operatorname{Sp}_2(\mathbb Z_{p^\alpha})=\SL_2(\mathbb Z_{p^\alpha})\), where the equality follows from the identity \(f(Av, Aw)=\det(A)f(v,w)\), for \(A\in \GL_2(\mathbb Z_{p^\alpha})\), and for all \(v,w\in \mathbb Z_{p^\alpha}^2\). In particular, the standard basis $e_1:=(1,0)^\top$ and $e_2:=(0,1)^\top$ satisfy \(f(e_1,e_2)=-1\). For each integer $0\le m\le\lfloor\alpha/2\rfloor$, we define \(L_m
:=p^m\mathbb Z_{p^\alpha}e_1
\oplus p^{\alpha-m}\mathbb Z_{p^\alpha}e_2
\subseteq\mathbb Z_{p^\alpha}^2\). Each $L_m$ is isotropic, because the symplectic pairing of any two of its elements is divisible by $p^\alpha$. Moreover, \(|L_m|=p^{\alpha-m}p^m=p^\alpha\), so Lemma~\ref{lem:symplectic-size} implies that $L_m$ is Lagrangian. We further notice that only $L_0$ is free, whereas those with $m>0$ are non-free. We summarize the classification that every Lagrangian submodule of $\mathbb Z_{p^\alpha}^2$ is symplectically equivalent to one of \(L_m\) in the following lemma (see also~\cite[Sec.~1]{Albouy_2009}).

\begin{lemma}
\label{lem:lagrangian-classification}
Consider the symplectic module $\mathbb Z_{p^\alpha}^2$, where $p$ is prime and $\alpha\in\mathbb N$, equipped with the standard symplectic form $f$. 
Then
\begin{enumerate}[(i)]
    \item Every Lagrangian submodule \(L\in\mathbb Z_{p^\alpha}^2\) is symplectically equivalent to some standard Lagrangian \(L_m\), that is, there exists a symplectic matrix \(S \in \operatorname{Sp}_2(\mathbb Z_{p^\alpha})\) such that \(L = S (L_m)\).
    \item A submodule \(L\subseteq \mathbb Z_{p^\alpha}^2\) is Lagrangian if and only if \(|L|=p^\alpha\).
    \item The total number of Lagrangian submodules of \(\mathbb Z_{p^\alpha}^2\) is \( |\operatorname{Lag}(\mathbb Z_{p^\alpha}^2)|=1+p+p^2+\cdots+p^\alpha.\)
\end{enumerate}
\end{lemma}
For completeness, we provide a short proof of Lemma~\ref{lem:lagrangian-classification} below. We use a standard consequence of the Smith normal form, which states that every submodule \(L\subseteq \mathbb Z_{p^\alpha}^2\) is generated by two generators \(p^mf_1\) and \(p^nf_2\), such that \(L_m=p^m\mathbb Z_{p^\alpha}f_1\oplus p^n\mathbb Z_{p^\alpha}f_2\); the reader is referred to \cite[Chapter~15]{brown1993matrices} for details on the Smith normal form.

\begin{proof}[Proof of Lemma.~\ref{lem:lagrangian-classification}]
(i) Let $L$ be a non-trivial submodule. We will first claim that $L$ has a direct sum decomposition of the form $L = p^m\mathbb Z_{p^\alpha}f_1\oplus p^n\mathbb Z_{p^\alpha}f_2$, where \(\{f_1,f_2\} \subseteq \mathbb Z_{p^\alpha}^2\) is a basis, and $0 \le m \le n \le \alpha$, with $m < \alpha$. By the Smith normal form, there exists a basis \(\{f'_1,f'_2\} \subseteq \mathbb Z_{p^\alpha}^2\) that generates \(L\), such that \(L\) has the form $L = d_1 \mathbb Z_{p^\alpha} f'_1 \oplus d_2 \mathbb Z_{p^\alpha} f'_2$, where $d_1, d_2 \in Z_{p^\alpha}$ with \(d_1 \mid d_2\), and $d_2$ can be zero. We note that every nonzero \(d\in\mathbb Z_{p^\alpha}\) can be written as \(d=p^kv\), for \(0\leq k<\alpha\), and \(v\in\mathbb Z_{p^\alpha}^\times\), and since multiplication by a unit preserves \(\mathbb Z_{p^\alpha}\), we have \(d\mathbb Z_{p^\alpha}=p^k\mathbb Z_{p^\alpha}\). From this, we may conclude that \(L=p^m\mathbb Z_{p^\alpha}f'_1\oplus p^n\mathbb Z_{p^\alpha}f'_2\), for some \(0\leq m < \alpha\), and \(0\leq n \le \alpha\), where \(n=\alpha\) handles the case that \(d_2=0\). Thus the claim holds automatically when $d_2=0$; so let us now assume $d_2 \neq 0$, in which case $m,n < \alpha$. If $m \le n$, then we are done, otherwise we simply write $L=p^n\mathbb Z_{p^\alpha}f'_2 \oplus p^m\mathbb Z_{p^\alpha}f'_1$, and this proves the claim. We note that this decomposition implies $|L|=p^{2\alpha -m - n}$. Next, by Lemma~\ref{lem:symplectic-size}, if \(L\) is additionally Lagrangian, then \(|L|=p^\alpha\), which implies \(m+n=\alpha\), and since $m \le n$, this gives $m \le \lfloor\alpha/2\rfloor$. Thus we have shown that any Lagrangian submodule \(L\) has the direct sum decomposition  $L = p^m \mathbb Z_{p^\alpha} f_1 \oplus p^{\alpha-m} \mathbb Z_{p^\alpha} f_2$, for $0\le m\le\lfloor\alpha/2\rfloor$, and $\{f_1, f_2\}$ a basis in \(\mathbb Z_d^2\).

Next, \(\{f_1,f_2\}\) generates the module \(\mathbb Z_d^2\), since \(\mathbb Z_d^2\) has rank $2$. Thus there exists $A \in \mathrm{GL}_2(\mathbb Z_{p^\alpha})$ such that \(Ae_j=f_j\), for \(j=1,2\), and this yields $A\cdot L_m = L$. To upgrade $A$ to a symplectic matrix, let $u = \det(A) \in \mathbb Z_{p^\alpha}^\times$, and we define the diagonal matrix $D = \begin{pmatrix} u^{-1} & 0 \\ 0 & 1 \end{pmatrix} \in \mathrm{GL}_2(\mathbb Z_{p^\alpha})$. As multiplication by the unit $u^{-1}$ is an automorphism of the ideal $p^m \mathbb Z_{p^\alpha}$, we observe that $D$ stabilizes $L_m$, that is, \[D\cdot L_m = p^m \mathbb Z_{p^\alpha} (u^{-1} e_1) \oplus p^{\alpha-m} \mathbb Z_{p^\alpha} e_2 =p^m \mathbb Z_{p^\alpha} e_1 \oplus p^{\alpha-m} \mathbb Z_{p^\alpha} e_2  = L_m.\]
Let us define $S = A D$, which gives $\det(S) = \det(A)\det(D) = u \cdot u^{-1} = 1$, and therefore we deduce that $S \in \mathrm{SL}_2(\mathbb Z_{p^\alpha}) = \mathrm{Sp}_2(\mathbb Z_{p^\alpha})$ . Furthermore, \(S(L_m) = A\cdot D\cdot L_m = A\cdot L_m = L\), showing that $L$ is symplectically equivalent to $L_m$.

(ii) If \(L\subseteq \mathbb Z_{p^\alpha}^2\) is Lagrangian, we get \(|L|=p^\alpha\) from Lemma~\ref{lem:symplectic-size}. Conversely, suppose we are given a non-trivial submodule \(L \subseteq R^{2k}\) satisfying \(|L|=p^\alpha\). From the proof of part (i), we have that $L = p^m\mathbb Z_{p^\alpha}f_1\oplus p^n\mathbb Z_{p^\alpha}f_2$, for $0 \le m \le n \le \alpha$, with $m \neq \alpha$, and thus $|L|=p^{2\alpha - m -n}$. This now implies $m+n=\alpha$, and the same proof as in part (i) gives us that \(L = p^m \mathbb Z_{p^\alpha} f_1 \oplus p^{\alpha-m} \mathbb Z_{p^\alpha} f_2=S(L_m)\). Because a subgroup remains Lagrangian under a symplectic action by Lemma~\ref{lem:symplectic-action-lagrangians}, \(L\) is indeed Lagrangian.

(iii) By the counting formula for isotropic lines in \(\mathbb Z_{p^\alpha}^2\), see \cite[Section~1]{Albouy_2009}, the number of Lagrangian subgroups of \(\mathbb Z_{p^\alpha}^2\) is
\(|\operatorname{Lag}(\mathbb Z_{p^\alpha}^2)|
=
1+p+p^2+\cdots+p^\alpha.
\)
\end{proof}

Throughout the remainder of the paper, we use the standard symplectic form.

\subsection{Modular Clifford group and Clifford hierarchy}
\label{ssec:modular_clifford}

We define the \textit{shift operator} \(X\) and the \textit{clock operator} \(Z\) by 
\[
    X |x\rangle := |x+1 \rangle,
    \qquad
    Z |x\rangle := \omega^x |x\rangle \quad \text{for } x \in \mathbb Z_d,
\]
where \(x+1\) is evaluated in \(\mathbb Z_d\). This implies the group commutator relation \(Z X = \omega X Z\), and both \(X\) and \(Z\) have order \(d\). 

Let us consider the $\mathbb Z_d$-module \(\mathbb Z_d^2\), with matrices in \(\M_2(\mathbb Z_d)\) acting on this module
by left multiplication. We equip \(\mathbb Z_d^2\) with the standard symplectic pairing $f$ turning it into the symplectic module \((\mathbb Z_d^2,f)\):
\[
f(v,w) := v^\top Jw,
\qquad
J := \begin{pmatrix}0&-1\\1&0\end{pmatrix}, \qquad v,w \in \mathbb Z_d^2.
\]
Thus, if \(v=(a_1,b_1)^\top\) and \(w=(a_2, b_2)^\top\), then \(f(v,w)= a_2 b_1-a_1 b_2\).

\begin{defn}[Pauli group]
\label{def:pauli_weyl}
For any $v =(a,b)^\top\in \mathbb Z_d^2$, we define the ordered Weyl operator as $W_v:= X^a Z^b$. The \emph{Heisenberg-Weyl Pauli group} $\mathcal{P}$ is defined as the group generated by these operators and all scalar phases
\[\mathcal P:=\{\lambda W_{v}: \lambda\in \mathbb T,\; v\in \mathbb Z_d^2\}\footnote{\label{fn:pauli-phase}One can equivalently use the finite phase convention \(\mathcal P_{2d}:=\{\tau^n W_v:n\in\mathbb Z,\ v\in \mathbb Z_d^2\}\), and \(\mathcal P_d:=\{\omega^n W_v:n\in\mathbb Z,\ v\in \mathbb Z_d^2\}\) in the $d$ even and odd cases, respectively, to define the Pauli groups. In both cases, their normalizers in \(\mathcal U(\mathcal H)\) is the same as that of
\(\mathcal P=\{\lambda W_v:\lambda\in\mathbb T,\ v\in \mathbb Z_d^2\}\), and hence gives the same
Clifford group, and the same Clifford hierarchy $\mathcal{C}_k$, for $k \ge 2$. By contrast, the smaller convention \(\mathcal P_d:=\{\omega^n W_v:n\in\mathbb Z,\ v\in \mathbb Z_d^2\}\), although still a group, can have a strictly smaller normalizer for $d$ even.}
.\]
The operators \(\{W_v:v\in \mathbb Z_d^2\}\) form an orthonormal basis of \(\B(\mathcal H)\), under the Hilbert-Schmidt inner product. Moreover, the group commutator relation between \(X\) and \(Z\) dictates the canonical commutation of the Weyl operators, that is,
\[
    W_vW_w
    =
    \omega^{f(v,w)}W_wW_v,
    \qquad
    \forall v,w\in \mathbb Z_d^2.
\]
Thus, $W_u$ and $W_v$ commute if and only if $f(u,v)=0$.
\end{defn}

An abelian subgroup of a group \(G\) is called \emph{maximal abelian} in \(G\) if it is not properly contained in any other abelian subgroup of \(G\). Given a submodule \(L\subseteq \mathbb Z_d^2\), we denote the Pauli subgroup associated with \(L\) as \(\mathcal P_L:=\{\,\lambda W_\ell:\lambda\in\mathbb T,\ \ell\in L\,\}\). Then \(L\) is isotropic if and only if \(\mathcal{P}_L\) is abelian. Moreover, \(L\) is Lagrangian exactly when \(\mathcal{P}_L\) is a maximal abelian Pauli subgroup, meaning that it is maximal abelian in \(\mathcal{P}\). We define the \textit{Pauli subalgebra} of $L$ as \(\mathfrak A_L:=\operatorname{span}_{\mathbb C}\{W_\ell:\ell\in L\}
\subseteq \B\mathcal{(H)}\). This is a unital \(*-\)subalgebra, and it is abelian exactly when \(L\) is isotropic. We call \(\mathfrak A_L\) a \emph{maximal abelian Pauli subalgebra}, if it is abelian and it is not properly contained in any other abelian Pauli subalgebra. More generally, an abelian subalgebra of \(\B(\mathcal{H})\) is called a \emph{maximal abelian subalgebra}, if it is not properly contained in any other abelian subalgebra of \(\B(\mathcal{H})\). For Pauli subalgebras, these two notions of maximality coincide. To see this, define the \textit{commutant} of \(\mathfrak A_L\) to be \(\mathfrak A_L':= \{ T \in \B(\mathcal{H}) : TA = AT, \; \forall A \in \mathfrak A_L\}\).  It is easy to see that \(\mathfrak A_L'\) is the subalgebra associated with the symplectic orthogonal complement of $L$, that is, \(\mathfrak A_L'=\mathfrak A_{L^\perp}\). This is because Weyl operators form a basis for \(\B(\mathcal{H})\) and the commutation relations imply that an operator commutes with every \(W_\ell\), for \(\ell\in L\), precisely when its Weyl expansion is supported on \(L^\perp\). Now suppose \(\mathfrak A_L\) is maximal among abelian Pauli subalgebras. Then \(L\) is isotropic. If \(v\in L^\perp\setminus L\), then the larger submodule \(L+\mathbb Z_d v\) is also isotropic, so its Pauli subalgebra is abelian and properly contains \(\mathfrak A_L\), contradicting our maximality assumption. Hence \(L=L^\perp\), and the equation \(\mathfrak A_L'=\mathfrak A_{L^\perp}\) gives \(\mathfrak A_L=\mathfrak A_L'\). Therefore, every operator commuting with \(\mathfrak A_L\) is already in it, proving maximality among all abelian subalgebras of \(\B(\mathcal{H})\). Conversely, for any Pauli subalgebra, maximality among all abelian subalgebras implies maximality among abelian Pauli subalgebras. Consequently, both maximal abelian notions for \(\mathfrak A_L\) are equivalent to \(L\) being Lagrangian.

\begin{defn}[Modular Clifford group and Clifford hierarchy]
\label{def:clifford-hierarchy}
The (modular) Clifford hierarchy is defined by setting \(\mathcal C_1 := \mathcal P\), and for \(k\in\mathbb N\) with \(k\geq 2\), we recursively define the $k^{\text{th}}$ level of the hierarchy as
\begin{equation*}
\mathcal C_k :=
    \{\, U\in \U(\mathcal H):
        U\mathcal P U^* \subseteq \mathcal C_{k-1}
    \,\}.
\end{equation*}
The elements of $\mathcal{C}_k$ are also known as the \emph{Clifford hierarchy gates}. The second level of the hierarchy \(\mathcal C_2\) forms a group and is called the \textit{modular Clifford group} (or simply the \textit{Clifford group}\footnote{The term \textit{Clifford group} for $d$ composite case is context dependent, but in this paper we will use this term to exclusively refer to the modular Clifford group, dropping the prefix `modular'.}), and its elements are called \emph{Clifford gates}. Since the conjugation action defining the hierarchy descends naturally to the projective unitary group \(\PU(\mathcal{H})\), we define the \textit{projective Clifford group}, denoted as \(\operatorname{P}\mathcal{C}_2:=\mathcal{C}_2/\mathbb T I\). For \(k>2\), the levels need not be groups, but they form a strictly increasing sequence \(\mathcal C_k \subset \mathcal C_{k+1}\). The inclusions
\(\mathcal C_k\subseteq\mathcal C_{k+1}\)
follow immediately by induction, while a construction is provided in Appendix~\ref{app:strict_nest}, which shows the strictness of the inclusions.
\end{defn}

Since the Clifford hierarchy is defined through conjugation of Pauli operators, we first record two elementary properties of unitary conjugation concerning projective order and maximal abelianness.

\begin{lemma}[Properties of unitary conjugation]
\label{lem:unitary-conjugation-prop}
    Let \(U\in\U(\mathcal{H})\). Then,
    \begin{enumerate}[(i)]
    \item For every subgroup \(G\subseteq\B(\mathcal{H})\), unitary conjugation by \(U\) induces an isomorphism  \(\Phi(V):=UVU^*\), where \(V\in G\). In particular, if \(G\subseteq\U(\mathcal{H})\), then for \(V\in G\), if the projective class \([V]\) has order \(t\), then so does \([UVU^*]\) in \(\operatorname{PU}(\mathcal{H})\).
    \item If \(G\) is a maximal abelian subgroup of a group \(K\subseteq\U(\mathcal H)\), then \(UGU^*\) is a maximal abelian subgroup in \(UKU^*\). Likewise, if \(\mathfrak A\) is a maximal abelian subalgebra in \(\B(\mathcal{ H})\), then so is \(U\mathfrak A U^*\).
    \end{enumerate}
\end{lemma}

\begin{proof}
    For part (i), if \(V,W\in G\), we have \(\Phi(VW)=UVWU^*=(UVU^*)( UWU^*)=\Phi(V)\Phi(W)\), so \(\Phi\) is a well-defined homomorphism. Conjugation by \(U^*\) induces its inverse, so \(\Phi\) is an isomorphism. If \(G\subseteq\U(\mathcal{H})\), then for every positive integer \(m\), we have \((UVU^*)^m=UV^mU^*\). Consequently, \((UVU^*)^m\) is a scalar multiple of the identity if and only if the same is true for \(V^m\). It follows from this that the two projective classes have the same order.
    
    For part (ii), conjugation by \(U\) preserves products, linear combinations, and commutation relations, and its inverse is conjugation by \(U^*\). Hence, the conjugated subgroup and subalgebra are abelian. Any strictly larger abelian subgroup of \(UKU^*\) containing \(UGU^*\) would conjugate back to a strictly larger abelian subgroup of \(K\) containing \(G\), contradicting maximality. The same argument applies to subalgebras of \(\B(\mathcal H)\).
\end{proof}

Next, we now relate the conjugation by Clifford gates to the symplectic action introduced in Lemma~\ref{lem:symplectic-action-lagrangians}. Similar to the case in \(\mathbb Z_{p^\alpha}\), the symplectic group \(\operatorname{Sp}_2(\mathbb Z_d)\) satisfies \(\operatorname{Sp}_2(\mathbb Z_d)=\operatorname{SL}_2(\mathbb Z_d)\). The following Clifford--symplectic correspondence is standard (see, for example, \cite[Theorem~4.1]{KorbelarTolar2023CliffordSemidirect}); for every \(C\in\mathcal C_2\), there exists a unique
\(\rho([C])\in\Sp_2(\mathbb Z_d)\), called the \emph{symplectic representative} of $C$, such that \([CW_vC^*]=[W_{\rho([C])v}]\), for all \(v\in\mathbb Z_d^2.\) Moreover, the resulting map \(\rho:\operatorname{P}\mathcal C_2\longrightarrow\Sp_2(\mathbb Z_d)\)is a surjective group homomorphism whose kernel is the projective Pauli subgroup. Thus there is a short exact sequence
\[
1\longrightarrow\mathbb Z_d^2
\longrightarrow\operatorname{P}\mathcal C_2
\xrightarrow{\rho}\Sp_2(\mathbb Z_d)
\longrightarrow1 \footnote{The existence of the homomorphism \(\rho\) should not be confused with the existence of a splitting. A semidirect product would require a group homomorphism \(s:\Sp_2(\mathbb Z_d)\to \operatorname{P\mathcal C}_2\) satisfying \(\rho\circ s=I\), which would yield a decomposition of the form \(\Sp_2(\mathbb Z_d)\ltimes \mathbb Z_d^2\).
Korbelář and Tolar show that such a splitting fails precisely when \(4\mid d\) \cite{KorbelarTolar2023CliffordSemidirect}.},
\]
where all maps are group homomorphisms. In the following lemma, we record some elementary structural properties of Clifford gates and their action on the Clifford hierarchy, which will be used in later sections.
\begin{lemma}
\label{lem:clif_prop}
The following properties hold for Clifford gates:
\begin{enumerate}[(i)]
\item For all \(k\geq2\), \(U \in \U(\mathcal{H})\), and \(C,\widehat C\in\mathcal C_2\), it holds that \(U\in\mathcal C_k\) if and only if \(CU\widehat C\in\mathcal C_k\).

\item Clifford conjugation preserves maximal abelian Pauli subgroups and subalgebras.

\item The permutation gates in \(\mathcal C_2\) are
precisely \(\Pi\ket{x}=\ket{a x+b}\), where \(a\in\mathbb Z_d^\times\) and \(b\in\mathbb Z_d\).
\end{enumerate}
\end{lemma}

\begin{proof}
For part (i), the inductive argument of \cite[ Proposition~3]{zeng_semi-clifford_2008} applies for every $d$.

For part (ii), for \(C\in\mathcal{C}_2\), we denote its symplectic representative as \(S=\rho([C])\). Therefore, for every submodule \(L\subseteq\mathbb Z_d^2\), we have \(C\mathcal P_LC^*=\mathcal P_{S(L)}\) and \(C\mathfrak A_LC^*=\mathfrak A_{S(L)}\). Since symplectic action preserves Lagrangian property from Lemma~\ref{lem:symplectic-action-lagrangians}, conjugation by Clifford gates also preserves the maximal abelian property for both the Pauli subgroup and the subalgebra. 

For part (iii), first suppose that \(\Pi\in\mathcal{C}_2\) is a permutation gate with \(\Pi\ket{x}=\ket{\pi(x)}\) for some permutation function \(\pi\). Then \(\Pi Z\Pi^*\ket{x} =\Pi Z\ket{\pi^{-1}(x)}=\omega^{\pi^{-1}(x)}\ket{x}\). Since \(\Pi\) is Clifford, \(\Pi Z\Pi^*\) is a diagonal Pauli, so there exist some \(\lambda\in\mathbb T\) and \(c_1\in\mathbb Z_d\) such that \(\Pi Z\Pi^*=\lambda Z^{c_1}\). Evaluating at \(x=0\) gives \(\lambda=\omega^c\), where \(c=\pi^{-1}(0)\). As conjugation does not change the projective order by Lemma~\ref{lem:unitary-conjugation-prop}(i), \(Z^{c_1}\) has projective order \(d\), so \(c_1\in\mathbb Z_d^\times\). Comparing diagonal entries gives
\(\pi^{-1}(x)=c_1x+c\) in \(\mathbb Z_d\), and hence \(\pi(x)=c_1^{-1}x-c_1^{-1}c.\) Thus \(\pi(x)=ax+b\), with
\(a=c_1^{-1}\in\mathbb Z_d^\times\) and
\(b=-c_1^{-1}c\in\mathbb Z_d\). Conversely, if \(\Pi\ket{x}=\ket{a x+b}\) with \(a\in\mathbb Z_d^\times\) and \(b\in\mathbb Z_d\), then a direct calculation gives \(\Pi X\Pi^*=X^a\) and \(\Pi Z\Pi^*=\omega^{-a^{-1}b}Z^{a^{-1}}.\) Thus \(\Pi\) is Clifford. This extends the prime-dimensional characterization in \cite[Theorem~1 and Lemma~6]{silva_clifford_2025} to every \(d\geq2\).
\end{proof}

To study hierarchy gates \(U\in\mathcal{C}_k\), it is natural to use an ordered pair \((A,B):=(UXU^*,UZU^*)\), because every Weyl operator, conjugated by \(U\), takes the form \(A^aB^b\) for \(a,b\in\mathbb Z_d\). As we will see in Lemma~\ref{lem:conjugate_pairs}, the map \([U]\mapsto(A,B)\) establishes a one-to-one correspondence between projective Clifford hierarchy gates and conjugate pairs, which we define next.

\begin{defn}[Conjugate pairs]
\label{def:conjugate-pairs}
A \emph{conjugate pair} is an ordered pair \((A,B)\in \U(\mathcal H) \times \U(\mathcal H)\), satisfying \(BA=\omega AB\) and the both have order \(d\). We denote the set of all conjugate pairs by \(\mathrm{Conj}(\mathcal H)\)\footnote{Our ordering reverses that of de~Silva~\cite{silva_efficient_2021} to align with our convention \(W_v=X^aZ^b\) for \(v=(a,b)^{\mathsf T}\). The associated pair then gives \(UW_vU^*=A^aB^b\), retaining the same phase-space labels and product order.}.
\end{defn}

The next lemma collects the basic properties of conjugate pairs. In particular, part~(ii) establishes the correspondence with the projective Clifford hierarchy gates and is a finite-dimensional form of the Stone--von Neumann theorem. These properties are standard in prime dimensions; see \cite[Lemmas 3--4 and Theorems 3--4]{silva_efficient_2021}. The proofs given there extend to arbitrary \(d\), as they do not rely on the field structure of \(\mathbb Z_d\), so we omit them here.

\begin{lemma}
\label{lem:conjugate_pairs}
The following properties hold:
\begin{enumerate}[(i)]
    \item Let \(A,B\in\U(\mathcal H)\) satisfy \(BA=\omega AB\). Then the operators \(\{A^aB^b:(a,b)^{\mathsf T}\in \mathbb Z_d^2\}\) form an orthogonal basis of \(\B(\mathcal H)\).

    \item If \((A,B)\) forms a conjugate pair, then there exists \(u_0\in\mathcal H\) of norm one, such that \(Bu_0=u_0\). For such a $u_0$, define
    \(u_x := A^x u_0\), for \(x\in \mathbb Z_d.\)
    Then \(\{u_x:x\in \mathbb Z_d\}\) is an orthonormal basis of \(\mathcal H\), and the unitary \(U\) defined by \(U|x\rangle := u_x\)
    satisfies \(UXU^* = A\), and \(UZU^* = B\). This \(U\) is unique up to a scalar phase.

    \item For a conjugate pair \((A,B)\), define \(\mathcal W(A,B):=\{\lambda A^aB^b:\lambda\in \mathbb T,\ (a,b)^{\mathsf T}\in \mathbb Z_d^2\}\). Then \(\mathcal W(A,B)\) is a group. Moreover, \(\mathcal W(A,B)\subseteq \mathcal C_k\) if and only if the unitary \(U\) from (ii) lies in \(\mathcal C_{k+1}\). 
\end{enumerate}
\end{lemma}
 
\subsection{Semi-Clifford gates and generalized semi-Clifford gates}
\label{ssec:semi_clifford_generalized}

\begin{defn}[Diagonal gates]
\label{def:diagonal-gates}
We define the group of diagonal unitaries \[\mathcal D:=\{\,D\in \U(\mathcal H): D|x\rangle=g(x)|x\rangle, \; \forall x\in\mathbb Z_d,\; \text{ for some } g:\mathbb Z_d \to\mathbb T \,\}.\] The diagonal gates at hierarchy level \(k\) are \(\mathcal D_k := \mathcal C_k\cap \mathcal D .\)
\end{defn}

Although \(\mathcal C_k\) is generally not a group for \(k\geq 3\), its diagonal subset \(\mathcal D_k\) is an abelian subgroup of \(\mathcal D\); see \cite[Theorem~1]{Cui_2017}.

\begin{defn}[Semi-Clifford gates]
\label{def:semi-clifford}
A gate \(U\) is called \emph{semi-Clifford} if there exist
\(C_1,C_2\in \mathcal C_2\) and \(D\in \mathcal D\) such that
\(U=C_1DC_2\). We denote the set of semi-Clifford gates as \(\mathcal{SC}\) and we define \(\mathcal{SC}_k:=\mathcal{C}_k\cap\mathcal{SC}\), which is the set of semi-Clifford gates in \(\mathcal{C}_k\).
\end{defn}
\begin{rmk}
The one-bit gate-teleportation protocol for qubits in~\cite{Zhou_2000} was extended to prime-dimensional qudits in~\cite[Section~5.1]{silva_efficient_2021}, which depends only on the decomposition \(U=C_1DC_2\) and extends unchanged to arbitrary \(d\). Thus, the semi-Clifford gates considered here admit the same resource-efficient teleportation protocol, including in composite
dimensions.
\end{rmk}

\begin{defn}[Permutation diagonal normal form]
\label{def:clifford-permutation}
A unitary \(U\in\U(\mathcal H)\) admits a
\emph{permutation diagonal normal form} if there exists a representation \(U=C_1\Pi D C_2\), where \(C_1,C_2\in\mathcal C_2\), \(D\in\mathcal D\), and \(\Pi\in \U(\mathcal{H})\) is a permutation of the computational basis states (also known as a \emph{permutation gate}). We denote the set of all such $U$ by \(\mathcal N\) and define \(\mathcal N_k:= \mathcal C_k \cap \mathcal N\).
\end{defn}
We will show later that both normal forms admit characterizations in terms of free Lagrangian submodules of \(\mathbb Z_d^2\). By Lemma~\ref {lem:symplectic-size}, the free Lagrangian submodules of \(\mathbb Z_d^2\) are precisely those generated by unimodular elements, that is, elements \(v=(a,b)^\top\in\mathbb Z_d^2\) satisfying \(\gcd(a,b,d)=1\). Equivalently, the map \(t\mapsto tv\) is an isomorphism of \(\mathbb Z_d\)-modules from \(\mathbb Z_d\) onto \(\mathbb Z_dv\), when $v$ is unimodular. Such elements are also called \emph{admissible} in the projective-ring literature; see~\cite{havlicek_projective_2008,Blunck_2000}. We define \(L_Z:=\mathbb Z_d \left( \begin{smallmatrix} 0 \\ 1 \end{smallmatrix}\right)\). Since \(W_{(0,1)^\top}=Z\), we can also write \(\mathfrak A_{L_Z} =\operatorname{span}_{\mathbb C}\{Z^b:b\in\mathbb Z_d\}\). The following lemma characterizes Weyl operators associated with a unimodular element as conjugation by Clifford gates on \(Z\) up to phase. The following lemma follows from the Pauli-Euclid-Gottesman lemma~\cite[Lemma~2]{farinholt2014ideal}:  
\begin{lemma}
\label{lem:unimodular-prop}
    An element $v\in\mathbb Z_d^2$ is unimodular if and only if there exist $C\in\mathcal C_2$ and $\lambda\in\mathbb T$ such that \(W_v=\lambda C^*ZC\). Consequently, for every free Lagrangian $L\subseteq\mathbb Z_d^2$, there exists $C\in\mathcal C_2$ such that $C\mathcal P_LC^*=\mathcal P_{L_Z}$, where \(\mathcal P_{L_Z}=\{\lambda Z^b:\lambda\in\mathbb T,\ b\in\mathbb Z_d\}\). 
\end{lemma}

\begin{proof}
For unimodular \(v=(a,b)^\top\in\mathbb Z_d^2\), there exist \(s,t\in\mathbb Z_d\) such that \(sa+tb \equiv 1\pmod d\). We define the matrix \(S=\left(\begin{smallmatrix}b&-a\\s&t\end{smallmatrix}\right)\), satisfying \(Sv=(0,1)^\top\). We note that \(\det (S)=1\), so \(S\in\operatorname{Sp}_2(\mathbb Z_d)\). By the surjectivity of \(\rho\), there exist $C\in\mathcal C_2$ and $\lambda\in\mathbb T$ such that \(\rho([C])=S\), and \(W_v=\lambda C^*ZC\). Conversely, \(W_v=\lambda C^*ZC\) implies that \(\rho([C])v=(0,1)^\top\). Since \(\rho([C])\) is invertible and invertible matrices preserve unimodularity, \(v\) is unimodular. 

To prove the second half of the lemma, given that \(L\subseteq \mathbb Z_d^2\) is a free Lagrangian submodule, by Lemma~\ref{lem:symplectic-size}, there exists a unimodular element \(v=(a,b)^\top\in\mathbb Z_d^2\) such that \(L=\mathbb Z_dv\). Then the first half of the lemma provides a Clifford gate \(C\) such that \(C\mathcal P_LC^*=\mathcal P_{L_Z}\).
\end{proof}

\begin{lemma}[Unimodular-element criteria]
\label{lem:unimodular-direction-criteria}
For $U\in\U(\mathcal H)$, the following criteria hold:
\begin{enumerate}[(i)]
    \item \(U\in\mathcal{SC}\) if and only if there exists an unimodular \(v \in \mathbb Z_d^2\) such that \(UW_vU^*\in\mathcal P\).
    \item \(U\in\mathcal N\) if and only if there exist unimodular \(v,w \in \mathbb Z_d^2\) such that \(UW_vU^*\) and \(W_w\) commute.
\end{enumerate}
\end{lemma}

\begin{proof}
We prove the two parts separately.
\begin{enumerate}[(i)]
\item If \(U\in\mathcal{SC}\), then \(U=C_1DC_2\) for some \(C_1,C_2\in\mathcal{C}_2\), and \(D\in\mathcal{D}\). Thus \(C_2^*ZC_2\) is proportional to \(W_v\), for some \(v=\rho([C_2])^{-1}(0,1)^\top\), which is unimodular by Lemma~\ref{lem:unimodular-prop}. Then \(U(C_2^*ZC_2)U^*=C_1ZC_1^*\) is Pauli, proving one direction of the lemma. Conversely, suppose \(UW_vU^*\in\mathcal P\) with \(v\) unimodular. Then \(UW_vU^*=\mu W_w\), for some \(\mu\in\mathbb T\) and \(w\in\mathbb Z_d^2\). Since conjugation preserves projective order by Lemma~\ref{lem:unitary-conjugation-prop}(i), it means that \(w\) is also unimodular. By Lemma~\ref{lem:unimodular-prop}, we can then choose Clifford gates \(C_v,C_w\) such that \(C_v\mathcal P_{\mathbb Z_dv}C_v^*
=\mathcal P_{L_Z}\), and \(C_w\mathcal P_{\mathbb Z_dw}C_w^* =\mathcal P_{L_Z}\). Now, the relation \(UW_vU^*=\mu W_w\) implies \(U\mathcal P_{\mathbb Z_dv}U^*=\mathcal P_{\mathbb Z_dw}\), and so a calculation shows that the unitary \(V:=C_wUC_v^*\) satisfies \(V\mathcal{P}_{L_Z}V^*=\mathcal{P}_{L_Z}\). In particular, we have \(VZV^*=\lambda Z^a\), for some \(\lambda\in\mathbb T\), and \(a\in\mathbb Z_d\), and since conjugation preserves projective order by Lemma~\ref{lem:unitary-conjugation-prop}, we further deduce \(d=d/\gcd(a,d)\), or equivalently \(a\in\mathbb Z_d^\times\). Also, taking \(d\)-th powers of the relation \(VZV^*=\lambda Z^a\) gives us \(\lambda^d=1\), which implies \(\lambda=\omega^b\) for some \(b\in\mathbb Z_d\). Finally, we construct a permutation gate \(\Pi|x\rangle:=|a^{-1}(x-b)\rangle\), which is Clifford by Lemma~\ref{lem:clif_prop}(iii). We can check that it also satisfies \(\Pi Z\Pi^*=\omega^bZ^a\), and therefore, the operator \(D:=\Pi^*V\) commutes with \(Z\), and is hence diagonal. We thus conclude from the definition of $V$ that \(U=(C_w^*\Pi)DC_v\in\mathcal{SC}\), because \(C_w,\Pi\in\mathcal{C}_2\) and \(\mathcal{C}_2\) is a group.
\vspace{0.1cm}
\item Let us first assume that \(U\in\mathcal{N}\). Then by definition, there exists \(C_1,C_2\in\mathcal{C}_2\), \(D\in\mathcal{D}\), and \(\Pi\) some permutation gate, such that \(U=C_1\Pi DC_2\). By Lemma~\ref{lem:unimodular-prop}, there exists unimodular \(v,w\in\mathbb Z_d^2\), such that \(W_v=\lambda C_2^*ZC_2\) and \(W_w =\mu C_1ZC_1^*\), for some \(\lambda,\mu\in\mathbb T\). Then \(UW_vU^*=\lambda C_1\Pi D\,Z\,D^*\Pi^*C_1^*=\lambda C_1 (\Pi Z\Pi^*) C_1^*\). Since \(\Pi Z\Pi^*\) simply permutes the diagonal entries of \(Z\), it is diagonal. A direct computation then shows that \(UW_vU^*\) and \(W_w\) commute.

Conversely, assume \(UW_vU^*\) commutes with \(W_w\), where \(v,w\in\mathbb Z_d^2\) are unimodular. We denote \(\widehat{W}:=UW_vU^*\), and we first note some properties of \(\widehat W\), \(W_v\), and \(W_w\). By Lemma~\ref{lem:unimodular-prop}, each of them are related to $Z$ by similarity transformations (up to some non-zero constant). Now, since $Z$ has $d$ distinct eigenvalues, so does each of \(\widehat W\), \(W_v\), and \(W_w\). Moreover, they are all unitary operators, and hence their eigenspaces corresponding to each distinct eigenvalue are orthogonal and one dimensional. In particular, since $\widehat{W}$ and $W_w$ commute, they are simultaneously diagonalizable, and hence share the same eigenbasis. Let the eigenvalues of $\widehat{W}$ (resp. $W_w$) be $\{\beta_1,\cdots,\beta_d\}$ (resp. $\{\beta'_1,\cdots,\beta'_d\}$). Then, by Lagrange's interpolation formula, there exists a complex polynomial $p$ such that $p(\beta_j)=\beta'_j$, for all $1\le j \le d$, using the distinctness of the eigenvalues; since the common eigenbasis of $\widehat{W}$ and $W_w$ is orthogonal, this then implies that $p(\widehat{W})=W_v$. Interchanging the roles of $\widehat{W}$ and $W_v$, we may also conclude that there exists a complex polynomial $q$ such that $q(W_v)=\widehat{W}$. 

Now, let us write \(L_1=\mathbb Z_dv\) and \(L_2=\mathbb Z_dw\); then from the previous discussion, we immediately get that \(U\mathfrak A_{L_1}U^*=\mathfrak A_{L_2}\). Also, by Lemma~\ref{lem:unimodular-prop}, there exist Clifford gates \(C_1,C_2\) such that \(C_1\mathcal P_{L_1}C_1^*
=\mathcal P_{L_Z}\), and \(C_2\mathcal P_{L_2}C_2^* =\mathcal P_{L_Z}\). We may then verify that \(T:=C_{2}UC_{1}^*\) normalizes \(\mathfrak A_{L_Z}\), that is, \(T\mathfrak A_{L_Z}T^*=C_{2}UC_{1}^*\mathfrak A_{L_Z}C_{1}U^*C_{2}^*=\mathfrak A_{L_Z}\); so, in particular, $TZT^\ast$ is diagonal. A direct computation now shows that for every \(j\in\mathbb Z_d\), we have \((TZT^*)T|j\rangle=TZ|j\rangle=\omega^jT|j\rangle\), and thus each $T|j\rangle$ is an eigenvector of $TZT^*$, and all the eigenvalues are distinct. Since $TZT^*$ is diagonal, we further conclude that $T|j\rangle$ must lie in the span of a computational basis state, for every $j$. Since the eigenvectors of $TZT^*$ are orthogonal (by distictness of the eigenvalues), this implies that $T=\Pi D$, for some permutation gate $\Pi$ and a diagonal unitary $D$. Hence, from the definition of $T$, we get \(U=C_{2}^*\Pi DC_{1}\in\mathcal N\). \qedhere
\end{enumerate}
\end{proof}
\begin{rmk}
\label{rmk:unimodular-hierarchy-gate}
    For \(U\in\mathcal{C}_k\), from the arguments above, we have \(U\in\mathcal{SC}_k\) if and only if there exists an unimodular \(v\) such that \(UW_vU^*\in\mathcal P\). Same holds for \(\mathcal{N}_k\), that is, for \(U\in\mathcal{C}_k\), we have \(U\in\mathcal{N}_k\) if and only if there exists unimodular elements \(v\) and \(w\) such that \(UW_vU^*\) and \(W_w\) commute.
\end{rmk}

In multi-qubit systems~\cite[Proposition~5]{zeng_semi-clifford_2008} and prime-dimensional multi-qudit systems~\cite[Theorem~7]{silva_efficient_2021}, the normal form \(U=C_1DC_2\) is equivalent to conjugating some maximal abelian Pauli subgroup onto a maximal abelian Pauli subgroup, that is, \(U\mathcal{P}_LU^*=\mathcal{P}_{\hat{L}}\), for some Lagrangian submodules \(L,\hat{L}\). This equivalence follows from the transitivity of the Clifford action on the associated Lagrangians, so every maximal abelian Pauli subgroup is conjugated by Clifford gates to the standard \(Z\)-type subgroup. Building on the qubit analysis of~\cite{Pllaha_2020}, Pauli-support methods in~\cite[Theorem~5]{silva_clifford_2025} characterize single-qudit semi-Clifford hierarchy gates in prime dimension by the containment of the Pauli supports of their conjugate pair in parallel lines. These characterizations do not extend automatically to composite dimensions, motivating the formulations in terms of Pauli subgroups and their linear spans introduced below.

\begin{defn}[Lagrangian semi-Clifford gates]
\label{def:lsc}
A unitary \(U\in\U(\mathcal H)\) is called \emph{Lagrangian semi-Clifford} (or a \emph{Lagrangian semi-Clifford gate}) if there exists a Lagrangian submodule \(L\subseteq\mathbb Z_d^2\) such that \(U\mathcal P_LU^*\subseteq\mathcal P\). We denote the set of Lagrangian semi-Clifford gates by \(\mathcal{LSC}\) and set \(\mathcal{LSC}_k:=\mathcal C_k \cap \mathcal{LSC}\).
\end{defn}

As semi-Cliffordness does not cover the full Clifford hierarchy in multi-qudit systems, \cite{zeng_semi-clifford_2008} extends semi-Clifford gates to generalized semi-Clifford gates that replace the subgroup conjugation with its linear span analog. For a single qudit in an arbitrary dimension, we formulate both conditions with weaker requirements on the conjugate object (Definition~\ref{def:lsc} and Definition~\ref{def:gsc}). At the subgroup level, we require only containment in the Pauli group; at the algebra level, we require only invariance under conjugation by Weyl operators. We show below, in Lemma~\ref{lem:pauli-conjugation-equivalence}, that these conditions already force the image to be a maximal abelian Pauli subgroup or the linear span of such a subgroup, respectively, thereby establishing equivalence with the stronger definition in previous literature (see, for example,~\cite {zeng_semi-clifford_2008,beigi2009c3semicliffordgeneralizedsemiclifford, marcus2026generalizedsemicliffordconjectureholds}).

\begin{defn}[Generalized semi-Clifford gates]
\label{def:gsc}
A unitary \(U\in\U(\mathcal H)\) is called \emph{generalized semi-Clifford} (or a \emph{generalized semi-Clifford gate}) if there exists a Lagrangian submodule \(L\subseteq\mathbb Z_d^2\) such that \(U\mathfrak A_LU^*\) is invariant under conjugation by every Weyl operator, that is,
\[
    W_v(U\mathfrak A_LU^*)W_v^*
    =U\mathfrak A_LU^*,
    \qquad\forall v\in\mathbb Z_d^2.
\]
We denote the set of generalized semi-Clifford gates by \(\mathcal{GSC}\) and set \(\mathcal{GSC}_k:=\mathcal C_k\cap\mathcal{GSC}\).
\end{defn}

\begin{lemma}
\label{lem:pauli-conjugation-equivalence}
For \(U\in\U(\mathcal{H})\), the following criteria hold:
\begin{enumerate}[(i)]
\item \(U\in\mathcal{LSC}\) if and only if \(U\mathcal P_LU^*=\mathcal P_{\hat L}\) for some Lagrangian submodules \(L,\hat L\). 
\item \(U\in\mathcal{GSC}\) if and only if \(U\mathfrak A_LU^*=\mathfrak A_{\hat L}\) for some Lagrangian submodules \(L,\hat L\).
\end{enumerate}
\end{lemma}
\begin{proof}
For the first equivalence, for some Lagrangian submodules \(L,\hat L\), the equality \(U\mathcal P_LU^*=\mathcal P_{\hat L}\) obviously implies \(U\in\mathcal{LSC}\) by definition. Conversely, given \(U\in\mathcal{LSC}\), there exists a Lagrangian submodule \(L\subseteq \mathbb Z_d^2\), such that \(U\mathcal P_LU^*\subseteq\mathcal{P}\) is an abelian Pauli subgroup, and so it equals \(\mathcal P_{\hat L}\) for some isotropic submodule \(\hat{L}\). Since by Lemma~\ref{lem:unitary-conjugation-prop}(i) unitary conjugation preserves the cardinality of the quotient of any Pauli subgroup modulo scalar phases, so we have that \(|\hat{L}|=|L|=d\). Lemma~\ref{lem:symplectic-size} therefore implies that \(\hat{L}\) is Lagrangian. 

For part (ii), assuming \(U\mathfrak A_LU^*=\mathfrak A_{\hat L}\), for some Lagrangian submodules \(L,\hat L\), we get \(U\in\mathcal{GSC}\) because every Pauli subalgebra is invariant under conjugation by Weyl operators. 

Conversely, we assume \(U\in\mathcal{GSC}\) and let \(L\) be the Lagrangian submodule that satisfies the definition of \(\mathcal{GSC}\). We write \(\mathfrak A:=U\mathfrak A_LU^*\). By assumption, \(\mathfrak A\) is invariant under conjugation by every Weyl operator, that is, \(W_uTW_u^*\in\mathfrak A\), for every \(T\in\mathfrak A\), and \(u\in\mathbb Z_d^2\). Fixing such a $T$, we may expand \(T\) uniquely in the Weyl operator basis to get \(T=\sum_{w\in \mathbb Z_d^2}a_wW_w\), with \(a_w\in\mathbb C\). Now for all \(v\in \mathbb Z_d^2\), we first claim that if \(a_v\neq 0\), then \(W_v\in\mathfrak A\). Using the expansion of $T$, we have  \(W_uTW_u^* =\sum_w a_w\omega^{f(u,w)}W_w\), for all $u \in \mathbb Z_d^2$. Then for any $v' \in \mathbb Z_d^2$, we evaluate the average
\[
\begin{aligned}
E_{v'}(T)
&:=
\frac1{d^2}\sum_{u\in \mathbb Z_d^2}
\omega^{-f(u,v')}W_uTW_u^*=
\sum_w a_w
\left(
\frac1{d^2}\sum_{u\in \mathbb Z_d^2}\omega^{f(u,w-v')}
\right)W_w = a_{v'} W_{v'}.
\end{aligned}
\] 
Since \(W_uTW_u^*\in\mathfrak A\) for every $u \in \mathbb Z_d^2$, we conclude $E_v(T) = a_v W_v \in \mathfrak A$, and this proves the claim. 

The above claim shows that $\mathfrak A$ contains at least one $W_u$, for some $u \in \mathbb Z_d^2$; thus if we define \(\hat{L}:=\{v\in\mathbb Z_d^2:W_v\in\mathfrak A\}\), then we know that it is non-empty. Clearly, from definition, we have $\mathfrak A_{\hat{L}} \subseteq \mathfrak A$. From the proof of the claim above, we also see that if $T=\sum_{w\in \mathbb Z_d^2}a_wW_w \in \mathfrak A$, then $W_v \in \mathfrak A_{\hat{L}}$ whenever $a_v \neq 0$, which implies that $T \in \mathfrak A_{\hat{L}}$. This shows $\mathfrak A \subseteq \mathfrak A_{\hat{L}}$, proving $\mathfrak A_{\hat{L}} = \mathfrak A$. It remains to show that $\hat{L}$ is Lagrangian. Since \(\mathfrak A_L\) is a maximal abelian subalgebra, it follows that \(U\mathfrak A_LU^* = \mathfrak A = \mathfrak A_{\hat{L}}\) is also a maximal abelian subalgebra, by Lemma~\ref{lem:unitary-conjugation-prop}(ii). But $\mathfrak A_{\hat{L}}$ is also an Pauli subalgebra, and so it is in fact a maximal abelian Pauli subalgebra, and hence $\hat{L}$ is Lagrangian.
\end{proof}

\begin{rmk}
    In either case, \(\hat L\) is uniquely determined by \(U\) and \(L\), since linear independence of Weyl operators implies both \(L\mapsto\mathcal{P}_L\) and \(L\mapsto\mathfrak A_L\) are injective.
\end{rmk}

Samce as the Clifford hierarchy, each of the four sequences \((\mathcal{SC}_k)_{k\geq1}\), \((\mathcal N_k)_{k\geq1}\), \((\mathcal{LSC}_k)_{k\geq1}\), and \((\mathcal{GSC}_k)_{k\geq1}\) is strictly increasing. The inclusions follow from the definitions and the nesting of the Clifford hierarchy. Strictness is shown  in Appendix~\ref{app:strict_nest}.

Moreover, the four types of gates, \(\mathcal{SC},\mathcal{LSC},\mathcal{GSC}\), and \(\mathcal{N}\), have interesting properties beyond those as merely hierarchy gates. For example, we exhibit an element in all these four classes that is not necessarily in any finite level of the Clifford hierarchy, and hence not in \((\mathcal{SC}_k)_{k\geq1}\), \((\mathcal N_k)_{k\geq1}\), \((\mathcal{LSC}_k)_{k\geq1}\), or \((\mathcal{GSC}_k)_{k\geq1}\), for any $k$. Already for a single qubit, consider \(D_\theta=\operatorname{diag}(1,e^{i\theta})\), where \(\theta/\pi\notin\mathbb Q\). Since \(D_\theta\) commutes with \(Z\), it fixes every element of the Lagrangian Pauli subgroup \(\mathcal P_{L_Z}\) under conjugation, where \(L_Z=\mathbb Z_d(0,1)^\top\). Thus \(D_\theta\mathcal P_{L_Z}D_\theta^*=\mathcal P_{L_Z}\), so \(D_\theta\in\mathcal{LSC}\). On the other hand, the identity \(D_\theta XD_\theta^*=e^{i\theta}XD_{-2\theta}\), together with the recursive definition of the hierarchy, gives \(D_\theta\in\mathcal C_k\) if and only if \(\theta\in\frac{\pi}{2^{k-1}}\mathbb Z\pmod{2\pi}\), for \(k\geq1\). Consequently, \(D_\theta\) lies outside every finite hierarchy level. This is the single qubit instance of the restriction to dyadic rational multiples of \(\pi\) in Anderson's normal form for semi-Clifford gates belonging to the hierarchy~\cite[Section~3]{Anderson_2024}. 

In Figure~\ref{fig:semi-clifford-venn}, we get two separate pictures for the four types of gates in multi-qubit or single-qudit of square-free dimensions, and a single qudit of dimension that is not square-free. Previously, for all \(k\in\mathbb N\), in the multi-qubit setting, it was known that \(\mathcal{SC}_k=\mathcal{LSC}_k\) and \(\mathcal{GSC}_k=\mathcal N_k\)~\cite[Proposition~1,~2]{zeng_semi-clifford_2008}; while in the multi-qudit setting with $d$ prime, it was known that \(\mathcal{SC}_k=\mathcal{LSC}_k\)~\cite[Theorem~7]{silva_efficient_2021}. More generally, the same equalities holds for a single qudit when $d$ is square-free, as we will show in Theorem~\ref{thm:square-free-class-equivalence}. On the contrary, when \(d\) is not square-free, there exist non-free Lagrangian submodules of $\mathbb Z_d^2$. Since the symplectic action preserves freeness of the Lagrangian submodules (Lemma~\ref{lem:symplectic-action-lagrangians}), these cannot be carried to the standard free Lagrangian through conjugation by Clifford gates. Consequently, the subgroup and algebra conditions defining \(\mathcal{LSC}\) and \(\mathcal{GSC}\) need not yield the normal forms defining \(\mathcal{SC}\) and \(\mathcal N\).

\begin{cor}
\label{cor:lsc-clifford-invariant}
For every dimension $d \in \mathbb N$, we have \(\mathcal{SC}\subseteq\mathcal{LSC}\cap\mathcal N\) and \(\mathcal{LSC}\cup\mathcal N\subseteq\mathcal{GSC}\). All four classes are closed under adjoints and left and right Clifford multiplication. Their intersections with \(\mathcal C_k\) are invariant under left and right multiplication by Clifford gates, for all \(k\geq2\).
\end{cor}

\begin{proof}
The inclusions are obvious from their definitions, Lemma~\ref{lem:unimodular-direction-criteria}, and Lemma~\ref{lem:clif_prop}. Adjoints and Clifford multiplication (both left and right) preserve both normal forms \(\mathcal{SC}\) and \(\mathcal{N}\). For the other two classes \(\mathcal{LSC}\) and \(\mathcal{GSC}\), Lemma~\ref{lem:pauli-conjugation-equivalence} proves closure under adjoints. For \(U\in\mathcal{LSC}\), we want to show that \(C_1UC_2\in\mathcal{LSC}\) for \(C_1,C_2\in\mathcal{C}_2\). By Lemma~\ref{lem:pauli-conjugation-equivalence}, there exist Lagrangian submodules \(L,\hat{L}\subseteq\mathbb Z_d^2\) such that \(U\mathcal P_LU^*=\mathcal P_{\hat L}\). Now let \(S_1,S_2\) be the symplectic representatives of \(C_1,C_2\) respectively. Then \((C_1UC_2)\mathcal P_{S_2^{-1}(L)}(C_1UC_2)^*=\mathcal P_{S_1(\hat L)}\), proving \(C_1UC_2\in\mathcal{LSC}\). Similarly, for \(U\in\mathcal{GSC}\), we get \((C_1UC_2)\mathfrak A_{S_2^{-1}(L)}(C_1UC_2)^*=\mathfrak A_{S_1(\hat L)}\).   This gives invariance under right and left multiplication by Clifford gates of \(\mathcal{GSC}\).
\end{proof}

\begin{figure}
  \centering
  \newcommand{\vennCapsule}[4]{%
    \draw[black, line width=1.15pt]
      ({#1-#3+#4},{#2-#4})
      -- ({#1+#3-#4},{#2-#4})
      arc[start angle=-90, end angle=90, radius=#4cm]
      -- ({#1-#3+#4},{#2+#4})
      arc[start angle=90, end angle=270, radius=#4cm]
      -- cycle;
  }

  \begin{minipage}[t]{0.485\textwidth}
    \centering
    \vspace{0.5em}
    \resizebox{\linewidth}{!}{%
      \begin{tikzpicture}[x=1cm,y=1cm,
                          every node/.style={text=black, font=\normalsize}]
        \vennCapsule{0}{0}{4.15}{2.45}
        \vennCapsule{0}{0.35}{2.15}{1.08}

        \node at (0,0.35) {$\mathcal{SC}=\mathcal{LSC}$};
        \node at (0,-2) {$\mathcal{GSC}=\mathcal N$};
      \end{tikzpicture}%
    }
    (a) Multi-qubit or a single qudit (\(d\) is square-free)\par
  \end{minipage}%
  \hfill
  \begin{minipage}[t]{0.485\textwidth}
    \centering
    \vspace{0.5em}
    \resizebox{\linewidth}{!}{%
      \begin{tikzpicture}[x=1cm,y=1cm,
                          every node/.style={text=black, font=\normalsize}]
        \vennCapsule{0}{0}{4.15}{2.45}
        \vennCapsule{-0.8}{0.22}{2.7}{1.55}
        \vennCapsule{0.8}{-0.20}{2.7}{1.55}
        \vennCapsule{0}{0.35}{1.48}{0.68}

        \node at (0,0.35) {$\mathcal{SC}$};
        \node at (-2.5,-0.72) {$\mathcal N$};
        \node at (2.5,-0.74) {$\mathcal{LSC}$};
        \node at (0,-2.02) {$\mathcal{GSC}$};
      \end{tikzpicture}%
    }
    (b) A single qudit (\(d\) is not square-free)\par
  \end{minipage}

  \caption{Relations among the four types of gates \(\mathcal{SC}\), \(\mathcal{LSC}\), \(\mathcal N\), and \(\mathcal{GSC}\).}
  \label{fig:semi-clifford-venn}
\end{figure}
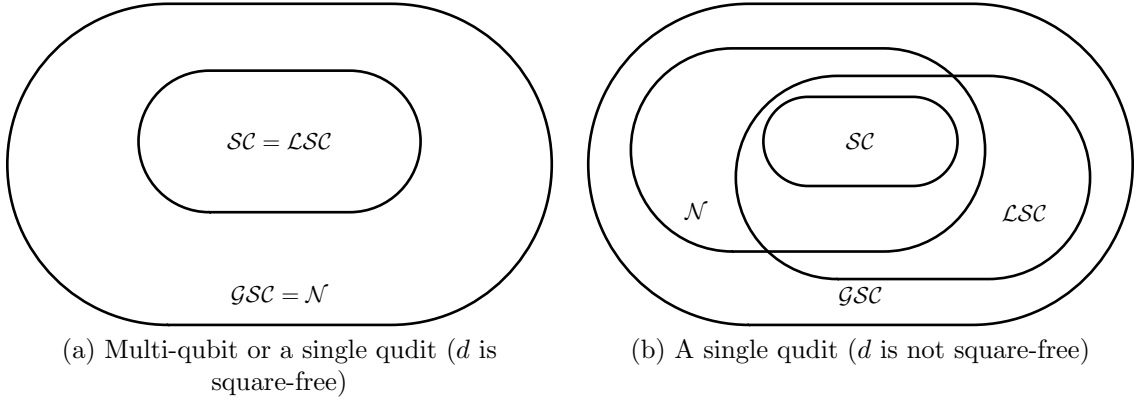
\section{Hierarchy gates are semi-Clifford when \texorpdfstring{$d$}{} is square-free}
\label{sec:square-free}

We recall prime-power factorization, that is, \(d = \prod_{j=1}^s q_j,\) where \(q_j = p_j^{\alpha_j},\) for \(1\le j\le s\) and the integers \(q_1,\dots,q_s\) are pairwise coprime. For each $1\leq j\leq s$, we define its corresponding ring, Hilbert space, and phase as $\mathbb{Z}_{q_j}:= \mathbb Z/q_j \mathbb Z$, $\mathcal{H}_j:= \mathbb C^{q_j}$, and $\omega_j:=e^{2\pi i/ q_j}$ respectively. 

We use the Chinese remainder theorem (CRT) to reduce the study of hierarchy gates to prime-power dimensions. For any decomposition \(d=\prod_{j=1}^s q_j\) into pairwise coprime factors, the associated CRT unitary transforms every \(d\)-dimensional hierarchy gate into a tensor product of hierarchy gates on the individual factors. When \(d\) is square-free, the known prime-dimensional semi-Clifford theorem~\cite{zeng_semi-clifford_2008,silva_clifford_2025} then applies to each factor. Finally, we use the CRT decomposition of Lagrangians to establish \(\mathcal{SC}=\mathcal{LSC}\) and \(\mathcal{GSC}=\mathcal N\) for arbitrary unitaries in square-free dimensions.

\subsection{CRT factorization}
\label{ssec:crt}
We recall that CRT gives a ring isomorphism \(\phi:\mathbb Z_d\longrightarrow\prod_{j=1}^s\mathbb Z_{q_j}\), defined by \(\phi(x):=(x\bmod q_1,\ldots,x\bmod q_s)\). We define \(\hat q_j:=d/q_j\), \(r_j:=\hat q_j^{-1}\pmod{q_j}\), and \(e_j:=r_j\hat q_j\pmod d\). Indeed, \(e_j\equiv1\pmod{q_j}\), and \(e_j\equiv0\pmod{q_k}\quad(k\ne j).\) So the inverse of \(\phi\) is \(\phi^{-1}(x_1,\ldots,x_s)=\sum_{j=1}^s x_j e_j \pmod d\). Thus this identification induces the unitary \(\Gamma\ket{x}:=\ket{\phi(x)}=\bigotimes_{j=1}^s\ket{x\bmod q_j}\). 

Let \(X_j,Z_j\) denote the standard Weyl generators on \(\mathcal H_j\). The factorization for Pauli groups is standard (see, for example, Looi and Griffiths \cite[Sec.~II.D, Thm.~4 and App.~A, Eqs.~(A4)--(A6)]{looi_tripartite_2011}). For \(x\in \mathbb Z_d\), write \(x_j := x\pmod q_j\in \mathbb Z_{q_j}\).
Then \(\Gamma |x\rangle = |x_1\rangle\otimes\cdots\otimes |x_s\rangle.
\)
Since addition by \(1\) modulo \(d\) becomes addition by \(1\) in each CRT coordinate, we get
\(\Gamma X|x\rangle=|x_1+1\rangle\otimes\cdots\otimes |x_s+1\rangle=
\left(\bigotimes_{j=1}^s X_{j}\right)\Gamma |x\rangle.
\)
Hence \(\Gamma X\Gamma^* =\bigotimes_{j=1}^s X_j\). For \(Z\), since we can write \(x=\sum_{j=1}^s x_j e_j\) in $\mathbb Z_d$, we get \(\omega^x
=
\prod_{j=1}^s \exp\left(\frac{2\pi i}{q_j}r_j x_j\right)
=
\prod_{j=1}^s \omega_j^{\,r_j x_j}.
\) Thus \(\Gamma Z\Gamma^* = \bigotimes_{j=1}^s Z_{j}^{\,r_j}\). Moreover, \(\Gamma X^aZ^b\Gamma^*
=(\Gamma X\Gamma^*)^a(\Gamma Z\Gamma^*)^b=\bigotimes_{j=1}^sX_j^{a_j}Z_j^{r_jb_j}.\)  We define \(\Theta:\mathbb Z_d^2\longrightarrow\prod_j\mathbb Z_{q_j}^2\) by \(\Theta(a,b)^\top:=\bigl((a\bmod q_j,\ r_jb\bmod q_j)\bigr)^\top_j\). Then the following Lemma is the phase-space counterpart of the decomposition on stabilizer groups; see \cite[Corollary~5]{looi_tripartite_2011}.

\begin{lemma}
\label{lem:crt-lag}
Let \(L\subseteq\mathbb Z_d^2\) be a submodule and we define \(L_j:=\operatorname{pr}_j(\Theta(L))\subseteq\mathbb Z_{q_j}^2\) for \(1\leq j\leq s.\) Then \(\Theta(L)=\prod_{j=1}^sL_j\) and \(\Theta(L^\perp)=\prod_{j=1}^sL_j^\perp.\) In particular, \(L\) is Lagrangian in \(\mathbb Z_d^2\) if and only if every \(L_j\) is Lagrangian in \(\mathbb Z_{q_j}^2\).
\end{lemma}
\begin{proof}
The inclusion \(\Theta(L)\subseteq\prod_jL_j\) is obvious, because by definition, every element of \(\Theta(L)\) has its \(j\)th component in \(L_j\). To prove \(\prod_jL_j\subseteq\Theta(L)\), we recall \(\hat q_j:=d/q_j\), \(r_j:=\hat q_j^{-1}\pmod{q_j}\), and \(e_j:=r_j\hat q_j\pmod d.\) For \(v\in L\), since \(L\) is a submodule, we know \(e_jv\) is in \(L\). As \(e_j\equiv1\pmod{q_j}\) and \(e_j\equiv0\pmod{q_k}\quad(k\ne j)\), we get \(\Theta(e_jv)
=(0,\ldots,0,\Theta_j(v),0,\ldots,0)\), that is, multiplication by \(e_j\) selects the \(j\)th CRT component of a submodule. Now choose any \((x_1,\ldots,x_s)\in\prod_jL_j\). For each \(j\), by the definition of \(L_j\), there exist some \(v^{(j)}\in L\) such that \(\Theta_j(v^{(j)})=x_j.\) Then we define \(v:=\sum_je_jv^{(j)}\in L\), and then \(\Theta(v)=(x_1,\ldots,x_s)\). This proves the reverse inclusion. We next prove the equivalence. We note that \([\Theta_j(v),\Theta_j(w)]=r_j([v,w]\bmod q_j)\). As each \(r_j\) is a unit, the vanishing of all \([v,w]\bmod q_j)\) is equivalent to the vanishing of \([\Theta_j(v),\Theta_j(w)]\). Therefore, we conclude that the statement is true.
\end{proof}

Since \(\Gamma\) maps Weyl operators to the tensor product and Pauli subalgebras are their complex linear spans, the factorization of Lagrangians immediately yields the corresponding factorization of Pauli subalgebras, \(\Gamma\mathfrak A_L\Gamma^*=\bigotimes_{j=1}^s\mathfrak A_{L_j}\), which will be used in Section~\ref{sec:gscg}. Meanwhile, it specializes in the known factorization of the Clifford group \cite[App.~B, Eqs.~(127)--(131)]{appleby_monomial_2012}. In the following, we will prove that this factorization extends to the entire Clifford hierarchy.  Before that, we should recall the notations and facts we use in the proof. The unsuffixed symbols refer to the ambient space \(\mathcal H\), while we denote \(\mathcal P(q_j)\), \(\mathcal C_k(q_j)\), and \(\mathcal D_k(q_j)\) to be the Pauli group, the \(k\)th Clifford hierarchy, and its diagonal sector on \(\mathcal H_j\), respectively. Besides, the proof uses two standard facts. First, the factors of a nonzero decomposable tensor are unique up to nonzero scalar rescalings whose product is one \cite[pp.~291--314]{Rom92}. The second one is that every automorphism of a full matrix algebra is inner by the Skolem--Noether theorem \cite[pp.~218--233]{Pie82}.
\begin{lemma}[CRT factorization of the hierarchy]
\label{lem:coprime-hierarchy-factorization}
For every \(k\in \mathbb N\),
\[
        \Gamma \mathcal C_k\Gamma^*
        =
        \left\{
        \bigotimes_{j=1}^s U_j:
         U_j\in \mathcal C_k(q_j)
        \right\}.
\]
\end{lemma}
\begin{proof}
The proof proceeds by induction on the hierarchy level $k$. For \(k=1\), since \(\Gamma X^aZ^b\Gamma^*=\bigotimes_{j=1}^s =X_j^{a_j}Z_j^{r_jb_j}\), where \(r_j=\hat q_j^{-1}\pmod{q_j}\) for each \(j\). Since every \(r_j\) is invertible, we conclude \(\Gamma\mathcal P(d)\Gamma^*=\{\bigotimes_{j=1}^sP_j:P_j\in\mathcal P(q_j)\}.\)

Assume the theorem holds for $k-1$ for \(k\ge 2\). First we prove the ``\(\supseteq\)'' direction. Let
\(U:=\bigotimes_{j=1}^s U_j\), where \(U_j\in\mathcal C_k(q_j)\). It suffices
to show that \(\Gamma^*U\Gamma\in\mathcal C_k\). For any Weyl operator
\(W\in\mathcal P\), \(k=1\) case gives \(\Gamma W\Gamma^*=\bigotimes_{j=1}^s W_j\) where \( W_j\in\mathcal P(q_j).\) Hence
\(U(\Gamma W\Gamma^*)U^*=
    \bigotimes_{j=1}^s U_jW_jU_j^*
    \in
    \bigotimes_{j=1}^s \mathcal C_{k-1}(q_j).
\)
By the inductive hypothesis, the last tensor product lies in \(\Gamma\mathcal C_{k-1}\Gamma^*\), so \((\Gamma^* U\Gamma) W(\Gamma^* U \Gamma)^*\in\mathcal{C}_{k-1}\). Therefore, \(\Gamma^*U\Gamma\in\mathcal C_{k}\), which is equivalent to \(U\in\Gamma \mathcal{C}_k\Gamma^*\).

Now to prove the other direction, consider \(U\in\mathcal{C}_k\) and define \(V:=\Gamma U \Gamma^*\in\U(\bigotimes_j\mathcal H_j)\). Denote the embedded local generators on \(\bigotimes_{j=1}^s \mathcal{H}_j\) as \(\widetilde X_j
:=
I_{\mathcal H_1}\otimes\cdots\otimes I_{\mathcal H_{j-1}}
\otimes X_{j}
\otimes I_{\mathcal H_{j+1}}\otimes\cdots\otimes I_{\mathcal H_s}\) and \(\widetilde Z_j:=I_{\mathcal H_1}\otimes\cdots\otimes I_{\mathcal H_{j-1}}
\otimes Z_{j}
\otimes I_{\mathcal H_{j+1}}\otimes\cdots\otimes I_{\mathcal H_s}\) with the corresponding pullbacks in \(\mathcal{P}\) as \(X^{(j)}:=\Gamma^*\widetilde X_j\Gamma\) and \(Z^{(j)}:=\Gamma^*\widetilde Z_j\Gamma\). By \(k=1\) case, both \(X^{(j)}\) and \(Z^{(j)}\) lie in \(\mathcal P\).  Now, consider conjugating \(\widetilde X_j\) and \(\widetilde Z_j\) by \(V\), we get $A_j := V\widetilde X_j V^*= \Gamma U X^{(j)} U^*\Gamma^*$ and $B_j := V\widetilde Z_j V^*=\Gamma U Z^{(j)} U^*\Gamma^*$; as $U$ is in the $k$-th level, both lie in $\Gamma \mathcal{C}_{k-1}\Gamma^*$. By the inductive hypothesis, $A_j$ and $B_j$ factorize completely into local tensor components, that is, \(A_j=\bigotimes_{l=1}^s A_{jl}\) and \(B_j=\bigotimes_{l=1}^s B_{jl}\), where \(A_{jl},B_{jl}\in \mathcal C_{k-1}(q_l).\)

We now fix \(j\) and next want to show that \(A_{jj}\) and \(B_{jj}\) satisfy the Weyl commutation relation with phase \(\omega_j\), and hence generate the full local matrix algebra \(\mathcal B(\mathcal H_j)\). We know that $\widetilde{Z}_j \widetilde{X}_j = \omega_{j} \widetilde{X}_j \widetilde{Z}_j$. Since this relation must be preserved by $A_j$ and $B_j$, there exist \(\lambda_{jl}\in\mathbb T\) such that \(B_{jl}A_{jl}=\lambda_{jl}A_{jl}B_{jl}\) for all \(l\) with \(\prod_{l=1}^s\lambda_{jl}=\omega_{j}.\) Meanwhile, as \(\det (B_{jl}) \det (A_{jl})=\lambda^{q_l}_{jl}\det (A_{jl}) \det (B_{jl})\), we have \(\lambda^{q_l}_{jl}=1\). Since \(\lambda_{j\ell}^{q_\ell}=1\) and \(\prod_\ell\lambda_{j\ell}=\omega_j\), whose order divides \(q_j\), coprimality forces \(\lambda_{jj}=\omega_j\) and \(\lambda_{j\ell}=1\) for \(\ell\ne j\). Thus by (i) of Lemma~\ref{lem:conjugate_pairs}, \(B_{jj}A_{jj}=\omega_jA_{jj}B_{jj}\) implies \(A_{jj}\) and \(B_{jj}\) generate the entire local matrix algebra $\B(H_j)$

We now want to show that the factors \(A_{j\ell}\) and \(B_{j\ell}\) are scalar multiples of the identity whenever \(\ell\ne j\). Consider \(\widetilde X_j, \widetilde Z_j, \widetilde X_r,\widetilde Z_r\) for \(r\neq j\). The corresponding \(A_j, B_j, A_r, B_r\) commutes accordingly. Since \(A_jA_r=A_rA_j\), equating the tensor components yields \(A_{j\ell}A_{r\ell}
=\mu_\ell A_{r\ell}A_{j\ell},
\) and \(\prod_\ell\mu_\ell=1.\) Taking determinants gives \(\mu_\ell^{q_\ell}=1\); then again the coprime argument implies every \(\mu_\ell=1\). The same argument is required for \(A_jB_r\), \(B_jA_r\), and \(B_jB_r\). So after taking \(l=j\), we find that for every \(r\ne j\), \(A_{rj}\) and \(B_{rj}\) commute with both \(A_{jj}\) and \(B_{jj}\). Since
\(A_{jj},B_{jj}\) form a conjugate pair on \(\mathcal H_j\), and the commutant of \(\{A_{jj},B_{jj}\}\) in \(\B(\mathcal H_j)\) consists only of scalars. Thus \(A_{rj}\) and \(B_{rj}\) are scalar whenever \(r\neq j\), and then we absorb their products into the \(j\)-th component, after which \(A_{jj}\) and \(B_{jj}\) are renamed. Therefore \(A_j=I\otimes\cdots\otimes A_{jj}\otimes\cdots\otimes I,\) and \(B_j= I\otimes\cdots\otimes B_{jj}\otimes\cdots\otimes I
\). Since \(\widetilde X_j^{q_j}=\widetilde Z_j^{q_j}=I\), conjugation by \(V\) gives \(A_{jj}^{q_j}=B_{jj}^{q_j}=I\). Together with \(B_{jj}A_{jj}=\omega_jA_{jj}B_{jj}\), this shows that \((A_{jj},B_{jj})\) is a conjugate pair on \(\mathcal H_j\). By Lemma~\ref{lem:conjugate_pairs}(ii), there exists \(V_j\in\mathrm U(\mathcal H_j)\) such that \(V_jX_jV_j^*=A_{jj}\) and \(V_jZ_jV_j^*=B_{jj}\). Set \(W=\bigotimes_{j=1}^sV_j\). Then \(V\) and \(W\) induce the same conjugation on every embedded generator \(\widetilde X_j,\widetilde Z_j\). Consequently, \(W^*V\) commutes with all these generators. Since they generate \(\B(\bigotimes_{j=1}^s\mathcal H_j)\), their common commutant consists only of scalar operators. Thus \(W^*V=\lambda I\) for some \(\lambda\in\mathbb T\), so
\(V=\lambda\bigotimes_{j=1}^sV_j\). Absorbing \(\lambda\) into \(V_1\) yields the required factorization.

Finally, we consider an arbitrary \(P_j\in\mathcal P(q_j)\), and define its embedded operator as \(\widetilde P_j:=I_{\mathcal{H}_1}\otimes\cdots\otimes P_j\otimes\cdots\otimes I_{\mathcal{H}_s}\). Therefore, \(V\widetilde P_jV^*=
I\otimes\cdots\otimes(V_jP_jV_j^*)\otimes\cdots\otimes I
\in\Gamma\mathcal C_{k-1}\Gamma^*.\) The induction hypothesis then gives \(V_jP_jV_j^*\in\mathcal C_{k-1}(q_j)\), proving \(V_j\in\mathcal C_k(q_j)\). This proves the inclusion ``\(\subseteq\)''. Taking all the above into account, the equality holds.
\end{proof}

\subsection{Square-free semi-Clifford theorem}
\label{ssec:square_free}

The CRT factorization reduces the square-free case to prime dimensions. We therefore recall the following known result, which combines the previously established qubit and odd-prime cases.

\begin{theorem}[Prime-dimensional semi-Clifford theorem
{\cite{zeng_semi-clifford_2008,silva_clifford_2025}}]
\label{thm:prime-semi-clifford}
We assume that the fixed ambient dimension is prime, that is, \(d=p\). Then, for every \(k\in\mathbb N\), \(\mathcal C_k=\mathcal{SC}_k\).
Equivalently, every \(U\in\mathcal C_k\) admits a decomposition \(U=C_1DC_2\), where \(C_1,C_2\in\mathcal C_2,\) and \(D\in\mathcal D_k\).
\end{theorem}
We now extend the prime-dimensional result to square-free composite dimensions.  By the CRT factorization of the Clifford hierarchy, the argument reduces to applying Theorem~\ref{thm:prime-semi-clifford} on each prime tensor factor.

\begin{lemma}
\label{lem:squarefree-semiclifford}
Let \(d\) be square-free. Then in a single qudit of \(d\) dimension, for every $k\in\mathbb N$, \(\Cliff_k=\SC_k\), that is, every single-qudit Clifford-hierarchy gate is semi-Clifford.
\end{lemma}

\begin{proof}
The direction $\SC_k\subseteq\Cliff_k$ is straightforward by the Clifford invariance in Lemma~\ref{lem:clif_prop}(i). We now prove the reverse inclusion. Let $U\in\Cliff_k(d)$. By Lemma~\ref{lem:coprime-hierarchy-factorization}, 
\(\Gamma U\Gamma^*=\bigotimes_{j=1}^sU_j\), where \(U_j\in\Cliff_k(p_j)\).
By the prime-dimensional \Cref{thm:prime-semi-clifford}, each local factor has a semi-Clifford decomposition
\(U_j=C_{j,1}D_jC_{j,2}\), where  \(C_{j,1}, C_{j,2}\in\Cliff_2(p_j)\) and \(D_j\in \mathcal D_k(p_j)\).
Thus, \(\Gamma U\Gamma ^*=\left(\bigotimes_j C_{j,1}\right)
\left(\bigotimes_j D_j\right)\left(\bigotimes_j C_{j,2}\right).\) Set \(C_1=\Gamma^*\left(\bigotimes_j C_{j,1}\right)\Gamma\), \(C_2=\Gamma^*\left(\bigotimes_j C_{j,2}\right)\Gamma\), and \(D=\Gamma^*\left(\bigotimes_j D_j\right)\Gamma.\)
By Lemma~\ref{lem:coprime-hierarchy-factorization} at $k=2$ and $q_i=p_i$, we have $C_1,C_2\in\Cliff_2$, and $D\in\mathcal C_k$.  

It remains to check that \(D\) is diagonal in the original computational basis. Since each \(D_j\) is diagonal in the standard basis of \(\mathcal H_j\), there exists a function \(\delta_j:\mathbb Z_{p_j}\to \mathbb T\) such that \(D_j\ket{x_j}=\delta_j(x_j)\ket{x_j}\), where\(x_j\in \mathbb Z_{p_j}\). Hence \(D\ket{x}=\Gamma^*\left(\bigotimes_{j=1}^s D_j\right)\Gamma\ket{x}=\left(\prod_{j=1}^s \delta_j(x_j)\right)\ket{x}.\) Thus, every computational basis element \(\ket{x}\) is an eigenvector of \(D\), so \(D\) is diagonal in the original computational basis of \(\mathcal H\). Therefore, \(U\in\SC_k\).
\end{proof}

We conclude with a second consequence of the CRT decomposition. When \(d\) is square-free, Lemma~\ref{lem:crt-lag} implies that every Lagrangian submodule of \(\mathbb Z_d^2\) is free of rank one.
The characterizations of \(\mathcal{LSC}\) and \(\mathcal{GSC}\) therefore recover the normal forms defining \(\mathcal{SC}\) and \(\mathcal N\), respectively. These equivalences hold for arbitrary unitaries, without a hierarchy restriction.

\begin{theorem}[Square-free dimension equivalences]
\label{thm:square-free-class-equivalence}
If \(d\geq2\) is square-free, then for a single qudit of dimension \(d\), we have \(\mathcal{SC}=\mathcal{LSC}\) and
\(\mathcal{GSC}=\mathcal N\).
\end{theorem}
\begin{proof}
The normal forms immediately give \(\mathcal{SC}\subseteq\mathcal{LSC}\) and \(\mathcal N\subseteq\mathcal{GSC}\), because diagonal gates centralize \(\mathcal{P}_Z:=\{\lambda Z^b: \lambda\in\mathbb T,\ b\in\mathbb Z_d\}\), while permutation--diagonal gates normalize its linear span, and the conjugation by Clifford gates preserves Pauli structure.

Before discussing its converse, we claim that every Lagrangian \(L\subseteq\mathbb Z_d^2\) is free. By Lemma~\ref{lem:crt-lag}, every Lagrangian \(L\subseteq\mathbb Z_d^2\) corresponds under CRT to a product of Lagrangians \(L_j\subseteq\mathbb Z_{p_j}^2\). Each \(L_j\) is one-dimensional, so choosing a generator in each factor and combining them by CRT gives a generator \(v\) of \(L\) for which \(r\mapsto rv\) is an isomorphism \(\mathbb Z_d\to L\). Hence \(L\) is free of rank one. Now let \(U\in\mathcal{LSC}\). There is a Lagrangian submodule \(L=\mathbb Z_dv\), with \(v\) unimodular, such that \(U\mathcal P_LU^*\subseteq\mathcal P\). In particular, \(UW_vU^*\in\mathcal P\), so Lemma~\ref{lem:unimodular-direction-criteria}(i) gives \(U\in\mathcal{SC}\). Consider \(U\in\mathcal{GSC}\), we can get \(U\mathfrak A_LU^*=\mathfrak A_{\hat{L}}\) for some Lagrangian submodules \(L=\mathbb Z_dv\) and \(\hat{L}=\mathbb Z_dw\), with \(v,w\) unimodular. Hence \(UW_vU^*\in\mathfrak A_{\hat{L}}\), so it commutes with \(W_w\). Lemma~\ref{lem:unimodular-direction-criteria}(ii) yields \(U\in\mathcal N\), completing both equalities.
\end{proof}
\section{Counterexample construction when \texorpdfstring{$d$}{} is not square-free}
\label{sec:non-square-free}
To prove the converse of the square-free result, for every non-square-free \(d\), we construct a gate \(U\in\mathcal C_3\) that is not semi-Clifford. From Lemma~\ref{lem:unimodular-direction-criteria}, we know that the Weyl labels whose conjugates remain in the Pauli group are important for testing the membership.
\begin{defn}[Pauli kernel]
\label{def:pauli_kernel}
For \(U\in\U(\mathcal H)\), its \emph{Pauli kernel} is \(K(U):=\{v\in\mathbb Z_d^2:UW_vU^*\in\mathcal P\}\).
\end{defn}
The Pauli Kernel is a subgroup of \(\mathbb Z_d^2\). Indeed, since \(W_vW_w^*\) is proportional to \(W_{v-w}\), the set \(K(U)\) contains zero and is closed under subtraction. The following lemma characterizes semi-Clifford and Lagrangian semi-Clifford gates in terms of \(K(U)\) and, at the third level, relates this kernel to the symplectic image of the conjugated Paulis. Its proof uses two standard facts about finite groups. Lagrange's theorem states that the order of a subgroup divides the order of the group, and Cauchy's theorem states that if a prime \(p\) divides the order of a finite group, then the group
contains an element of order \(p\).

\begin{lemma}[Pauli-kernel criteria]
\label{lem:pauli-kernel-criteria}
Let \(U\in\mathcal C_k\).
\begin{enumerate}[(i)]
\item \(U\in\mathcal{SC}\) if and only if \(K(U)\) contains an unimodular element. In particular, if \(d\) has a divisor \(t>1\) such that \(K(U)\subseteq t\mathbb Z_d^2\), then \(U\notin\mathcal{SC}_k\).
\item \(U\in\mathcal{LSC}\) if and only if \(K(U)\) contains a Lagrangian submodule. For \(d=p^\alpha\), where \(p\) is prime and \(\alpha\in\mathbb N\) , \(U\in\mathcal{LSC}\) is equivalent to \(|K(U)|\geq p^\alpha\).
\item If \(U\in\mathcal C_3\), the map \(\Phi_U:\mathbb Z_d^2\to\operatorname{Sp}_2(\mathbb Z_d)\), defined by \(\Phi_U(v):=\rho([UW_vU^*])\), is a homomorphism with kernel \(K(U)\). Its symplectic image \(\mathcal S_U\) is abelian and satisfies \(\mathcal S_U\cong\mathbb Z_d^2/K(U)\).
\end{enumerate}
\end{lemma}

\begin{proof}
Part~(i) follows from
Lemma~\ref{lem:unimodular-direction-criteria}(i).
The final assertion holds because no element in \(t\mathbb Z_d^2\) is unimodular when \(t\) divides \(d\).

For part (ii), maximal abelian Pauli subgroups are precisely the groups $\mathcal P_L$ with $L$ Lagrangian. By the definition of $K(U)$, \(U\mathcal P_LU^*\subseteq\mathcal P\) if and only if \(L\subseteq K(U)\). Thus the first equivalence follows from the definition of $\mathcal{LSC}$. Now suppose \(d=p^\alpha\). Every Lagrangian submodule has order $d$ by Lemma~\ref{lem:symplectic-size}. Hence, if $K(U)$ contains a Lagrangian submodule, then $|K(U)|\ge d$. Conversely, suppose $|K(U)|\ge p^\alpha$. We use the elementary fact that a finite abelian $p$-group $G$ with $|G|\ge p^r$ contains a subgroup of order $p^r$. We will prove that by induction on $r$. The case $r=0$ is immediate. For $r>0$, Cauchy's theorem gives a subgroup $N\le G$ of order $p$. The quotient $G/N$ is again an abelian $p$-group and has order at least $p^{r-1}$. By induction, it contains a subgroup of order $p^{r-1}$, whose inverse image in $G$ has order $p^r$. Since $K(U)$ is a finite abelian $p$-group, this fact provides a subgroup $L\subseteq K(U)$ of order $p^\alpha=d$. As an additive subgroup of $\mathbb Z_d^2$, $L$ is also a $\mathbb Z_d$-submodule. By Lemma~\ref{lem:lagrangian-classification}(ii), such a submodule is Lagrangian. This proves the converse.

For part (iii), the assumption \(U\in\mathcal C_3\) ensures that \(\Phi_U\) is well-defined. The Weyl multiplication law gives \([UW_{v+w}U^*]=[UW_vU^*][UW_wU^*]\). Since \(\rho\) is a homomorphism, so is \(\Phi_U\). Moreover, \(\ker\rho\) is the projective Pauli subgroup, hence \(\ker\Phi_U=K(U)\). The first isomorphism theorem gives \(\mathcal S_U\cong\mathbb Z_d^2/K(U)\). Thus \(\mathcal S_U\) is abelian and, when \(d=p^\alpha\), is a \(p\)-group.
\end{proof}

\subsection{The motivating example in dimension \texorpdfstring{\(9\)}{}}
\label{ssec:d9-prototype}
Throughout this subsection, we set $\eta:=e^{2\pi i/27}$, $\omega:=\eta^3$, and $\zeta:=\eta^9$. In dimension $9$, the unique nonfree Lagrangian submodule is $3\mathbb{Z}_9^2$, while the free Lagrangian submodules are precisely $\langle(3h,1)\rangle$ for $h\in\mathbb{Z}_3$ and $\langle(1,j)\rangle$ for $j\in\mathbb{Z}_9$. Because these submodules are pairwise distinct, Lemma~\ref{lem:lagrangian-classification}(iii) states that there are exactly $1+3+9=13$ Lagrangian submodules in $\mathbb{Z}_9^2$, confirming that our list is exhaustive. Building on the framework from Section~\ref{ssec:semi_clifford_generalized}, the remainder of this subsection establishes the precise relations among the four structural types in the Clifford hierarchy for a single $9$-dimensional qudit (summarized in Figure~\ref{fig:dim9-relations}). 

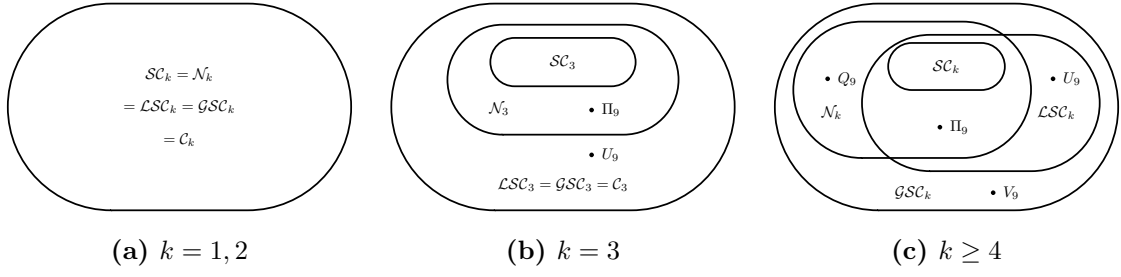
\begin{figure}
\centering
\newcommand{\dimNineCapsule}[4]{%
  \draw[black, line width=1.15pt]
    ({#1-#3+#4},{#2-#4})
    -- ({#1+#3-#4},{#2-#4})
    arc[start angle=-90, end angle=90, radius=#4cm]
    -- ({#1-#3+#4},{#2+#4})
    arc[start angle=90, end angle=270, radius=#4cm]
    -- cycle;
}

\begin{minipage}[t]{0.32\textwidth}
  \centering
  \vspace{0pt}
  \resizebox{\linewidth}{!}{%
    \begin{tikzpicture}[x=1cm,y=1cm,
                        every node/.style={text=black, font=\normalsize, align=center}]
      \path[use as bounding box] (-4.10,-2.50) rectangle (4.10,2.55);
      \dimNineCapsule{0}{0.10}{3.90}{2.35}
      \node at (0,0.10) {
        $\mathcal{SC}_k=\mathcal N_k$\\[9pt]
        $=\mathcal{LSC}_k=\mathcal{GSC}_k$\\[9pt]
        $=\mathcal C_k$
      };
    \end{tikzpicture}%
  }
  \par\smallskip
  \textbf{(a) $k=1,2$}\par
\end{minipage}%
\hfill
\begin{minipage}[t]{0.32\textwidth}
  \centering
  \vspace{0pt}
  \resizebox{\linewidth}{!}{%
    \begin{tikzpicture}[x=1cm,y=1cm,
                        every node/.style={text=black, font=\normalsize, align=center}]
      \path[use as bounding box] (-4.10,-2.50) rectangle (4.10,2.55);
      \dimNineCapsule{0}{0.10}{3.90}{2.35}
      \dimNineCapsule{0}{0.72}{2.625}{1.25}
      \dimNineCapsule{0}{1.12}{1.65}{0.55}

      \node at (0,1.12) {$\mathcal{SC}_3$};
      \node at (-1.45,0.06) {$\mathcal N_3$};
      \node at (0,-1.67)
        {$\mathcal{LSC}_3=\mathcal{GSC}_3=\mathcal C_3$};

      \fill[black] (0.65,0.03) circle (1.5pt)
        node[right=3pt] {$\Pi_9$};
      \fill[black] (0.65,-0.99) circle (1.5pt)
        node[right=3pt] {$U_9$};
    \end{tikzpicture}%
  }
  \par\smallskip
  \textbf{(b) $k=3$}\par
\end{minipage}%
\hfill
\begin{minipage}[t]{0.32\textwidth}
  \centering
  \vspace{0pt}
  \resizebox{\linewidth}{!}{%
    \begin{tikzpicture}[x=1cm,y=1cm,
                        every node/.style={text=black, font=\normalsize, align=center}]
      \path[use as bounding box] (-4.10,-2.50) rectangle (4.10,2.55);
      \dimNineCapsule{0}{0.10}{3.90}{2.35}
      \dimNineCapsule{-0.78}{0.49}{2.70}{1.55}
      \dimNineCapsule{0.78}{0.19}{2.70}{1.55}
      \dimNineCapsule{0}{1.02}{1.325}{0.535}

      \node at (0,1.02) {$\mathcal{SC}_k$};
      \node at (-2.60,-0.03) {$\mathcal N_k$};
      \node at (2.50,-0.03) {$\mathcal{LSC}_k$};
      \node at (-0.75,-1.85) {$\mathcal{GSC}_k$};

      \fill[black] (-2.70,0.74) circle (1.5pt)
        node[right=3pt] {$Q_9$};
      \fill[black] (2.42,0.74) circle (1.5pt)
        node[right=3pt] {$U_9$};
      \fill[black] (-0.15,-0.36) circle (1.5pt)
        node[right=3pt] {$\Pi_9$};
      \fill[black] (1.06,-1.85) circle (1.5pt)
        node[right=3pt] {$V_9$};
    \end{tikzpicture}%
  }
  \par\smallskip
  \textbf{(c) $k\geq4$}\par
\end{minipage}
\caption{Relations among the four types for a single qudit of dimension nine. (a) At levels one and two, all four classes coincide with the corresponding hierarchy level. (b) At level three, $\mathcal{SC}_3\subsetneq\mathcal N_3\subsetneq\mathcal{LSC}_3
  =\mathcal{GSC}_3=\mathcal C_3$. (c) For every $k\geq4$, the gates $\Pi_9,U_9,Q_9,V_9$
  realize the four regions determined by these classes.}
\label{fig:dim9-relations}
\end{figure}

We first construct a permutation gate \(\Pi_9\), which cycles \(|0\rangle,|3\rangle,|6\rangle\) and fixes the other six computational-basis states. We will prove \(\Pi_9\in(\mathcal N_3\cap\mathcal{LSC}_3)\setminus\mathcal{SC}_3\) in Appendix~\ref{ssec:sc_3}, since the argument is technical and is not needed for the subsequent proofs. Moreover, the third-level hierarchy exhibits certain special properties.
\begin{prop}
\label{prop:c3-lsc-gsc-d9}
For a single qudit of dimension $9$, the following relations hold at the third level of the Clifford hierarchy:
\begin{enumerate}[(i)]
\item \(\mathcal C_3=\mathcal{LSC}_3=\mathcal{GSC}_3\)
\item \(\mathcal N_3\subseteq\mathcal{LSC}_3\).
\end{enumerate}
\end{prop}
\begin{proof}
By Lemma~\ref{lem:pauli-kernel-criteria}(ii)--(iii), it suffices to prove \(|\mathcal S_U|\leq9\), since \(|K(U)|\,|\mathcal S_U|=81\). The argument is given in Appendix~\ref{ssec:sc_3}. It yields \(\mathcal C_3\subseteq\mathcal{LSC}_3\); Corollary~\ref{cor:lsc-clifford-invariant} then gives part (i). Part~(ii) follows immediately because \(\mathcal N_3\subseteq\mathcal C_3\).
\end{proof}

Moreover, the inclusion \(\mathcal N_3\subseteq\mathcal{LSC}_3\) is strict. To show this, we construct another gate \(U_9\in\mathcal{LSC}_k\setminus\mathcal N_k\) for every \(k\geq3\). We later extend this construction
to obtain a third-level gate that is not semi-Clifford in every non-square-free dimension. First, we define the Clifford gates \(E_9|x\rangle := \omega^{x(x-1)/2}|x\rangle\) and \(H_9|x\rangle := \frac{1}{3}\sum_{y\in \mathbb Z_9}\omega^{xy}|y\rangle.\) We note that \(x(x-1)/2\) is always an integer, so we do not need a finer phase. These are respectively the standard phase-shift and the discrete Fourier transform gates, which normalize the Weyl--Pauli group and hence lie in \(\mathcal C_2\); see \cite[Eqs.~(13)--(18) and Proposition~1]{farinholt2014ideal}. We define \(A_9 := Z^*H_9^*E_9^{-3}H_9\), and \(B_9 := XE_9^3\); since \(\mathcal{C}_2\) is a group, they all belong to \(\mathcal{C}_2\). Moreover, the Pauli factors \(X\) and \(Z^*\) do not affect the symplectic representation. A direct calculation gives
\[\rho([A_9]) = \begin{pmatrix}1 & 3 \\ 0 & 1\end{pmatrix}, \quad \rho([B_9]) = \begin{pmatrix}1 & 0 \\ 3 & 1\end{pmatrix}.
\]
We will prove that \((A_9,B_9)\) forms a conjugate pair and then, by the discrete von Neumann theorem (Lemma~\ref{lem:conjugate_pairs}), we can guarantee the existence of the corresponding third-level hierarchy gate \(U_9\), which is not semi-Clifford. Details are shown below.

\begin{prop}\label{prop:d9-non-semi-clifford}
The unitary gate \(U_9\) constructed above lies in \(\mathcal{LSC}_3/\mathcal{N}_3\), and thus it is not semi-Clifford.
\end{prop}
\begin{proof}
A direct calculation gives \(A_9^3=\zeta^{-1}Z^{-3}\) and \(B_9^3=\zeta X^3\). Consequently, conjugation by \(U_9\) maps the Pauli subgroup associated with the nonfree Lagrangian
\(3\mathbb Z_9^2\) onto itself. Thus \(U_9\in\mathcal{LSC}_3\).

Suppose \(U_9\in\mathcal N_3\). By Lemma~\ref{lem:unimodular-direction-criteria}(ii), there are addmisible \(v=(a,b)^{\mathsf T}\) and \(w=(c,e)^{\mathsf T}\) in \(\mathbb Z_9^2\) such that \(G:=U_9W_vU_9^*=A_9^aB_9^b\) commutes with \(W_w\). On the one hand, the symplectic representation of \(G\) is
\[
\rho([G])
=
\begin{pmatrix}1&3\\0&1\end{pmatrix}^{a}
\begin{pmatrix}1&0\\3&1\end{pmatrix}^{b}
=
\begin{pmatrix}1&3a\\3b&1\end{pmatrix}
\quad\text{over }\mathbb Z_9.
\]
Commutativity of \(G\) and \(W_w\) implies \(ae\equiv bc\equiv0\pmod3\). On the other hand, the identities for \(A_9^3\) and \(B_9^3\), together with \(B_9A_9=\omega A_9B_9\), give \([G^3]=[X^{3b}Z^{-3a}]\). Since \(G^3\) also commutes with \(W_w=X^cZ^e\),
the Weyl commutation relation gives
\(\omega^{-3(ac+be)}=1\), and hence
\(ac+be\equiv0\pmod3\). Meanwhile, admissibility of \(v\) means that at least one of
\(a,b\) is nonzero modulo \(3\).
If \(a\not\equiv0\pmod3\), then \(ae\equiv0\)
forces \(e\equiv0\), and \(ac+be\equiv0\)
then forces \(c\equiv0\).
If \(b\not\equiv0\pmod3\), the same equations
give \(c\equiv e\equiv0\).
In either case, \(w\) is not unimodular,
a contradiction. Therefore \(U_9\notin\mathcal N_3\).
Since \(\mathcal{SC}_3\subseteq\mathcal N_3\),
\(U_9\) is not semi-Clifford.
\end{proof}

Moreover, since \(\Pi_9\) lies exactly in \(\mathcal{C}_3\), it follows that \(\Pi_9\in(\mathcal{LSC}_k\cap\mathcal{N}_k)/\mathcal{SC}_k\) and \(U_9\in\mathcal{LSC}_k\setminus\mathcal N_k\)
for every \(k\geq3\). However, the inclusion \(\mathcal N_3\subseteq\mathcal{LSC}_3\) and equality \(\mathcal C_3=\mathcal{LSC}_3=\mathcal{GSC}_3\) does not extend to higher levels. In Appendix~\ref{app:dim9}, we construct two gates \(Q_9\) and \(V_9\) satisfying  \(Q_9\in\mathcal N_k\setminus\mathcal{LSC}_k\) and \(V_9\in\mathcal{GSC}_k/(\mathcal{LSC}_k\cup\mathcal{N}_k)\) 
for every \(k\geq4\). 

\subsection{Third-level gates that are not semi-Clifford in arbitrary non-square-free dimensions}
\label{ssec:not-sc-non-square-free}

In this subsection, we consider the dimension \(d=\prod_{j=1}^s p_j^{\alpha_j}\) to be non-square-free, so that \(\alpha_j\ge 2\) for some \(j\). We now construct a general family of $\mathcal{C}_3$ gates that are not semi-Clifford by observing that the factor \(3\) in the previous \(9\)-dimensional construction serves as a specific instance of a zero-divisor satisfying the following properties.

\begin{lemma}
\label{lem:square-zero-divisors}
An integer \(d\geq2\) is non-square-free if and only if it has a divisor \(r\) with \(1<r<d\) and \(r^2\equiv0\pmod d\). One may take \(r:=\prod_jp_j^{\lceil\alpha_j/2\rceil}\) and define a divisor \(t:=d/r=\prod_jp_j^{\lfloor\alpha_j/2\rfloor}>1\). Then \(ra\equiv0\pmod d\) if and only if \(a\in t\mathbb Z_d\).
\end{lemma}
\begin{proof}
For \(r=\prod_jp_j^{\beta_j}\mid d\), the condition \(d\mid r^2\) is equivalent to \(2\beta_j\geq\alpha_j\) for every \(j\). Such an \(r\) can be proper exactly when some \(\alpha_j\geq2\), and the displayed choice then works. The remaining assertions follow from \(d=rt\) and \(d\mid r^2\).
\end{proof}
\begin{lemma}\label{lem:explicit-clifford-shears}
Let \(r,t\) be as in Lemma~\ref{lem:square-zero-divisors}. There is a conjugate pair of Clifford gates \((A,B)\) with
\[
\rho([A])=
\begin{pmatrix}
1&r\\
0&1
\end{pmatrix},
\qquad
\rho([B])=
\begin{pmatrix}
1&0\\
r&1
\end{pmatrix}.
\]
Moreover, the conjugate pair \((A,B)\) determines a unitary \(U\in\mathcal C_3\), whose Pauli kernel is \(K(U)=t\mathbb Z_d^2\). In particular, \(U\) is not semi-Clifford.
\end{lemma}
\begin{proof}
We define
\[
H|x\rangle=\frac1{\sqrt d}\sum_{y\in\mathbb Z_d}\omega^{xy}|y\rangle,
\qquad E_r|x\rangle=\omega^{rx(x-1)/2+\epsilon x}|x\rangle,
\qquad
\epsilon=\begin{cases}0,&t\text{ odd},\\r/2,&t\text{ even}.\end{cases}
\]
Since \(t\mid r\), \(r\) is even whenever \(t\) is even, so the exponent is integral. Its change under \(x\mapsto x+t\) is
\(rtx+rt(t-1)/2+\epsilon t\), a multiple of \(d\) in both parity cases. Thus the phases are \(t\)-periodic, in particular well-defined on \(\mathbb Z_d\), and \(E_r\) is a polynomial in \(Z^r\).

The adjacent phase ratio gives \(E_rXE_r^*=\omega^\epsilon XZ^r\), while \(E_rZE_r^*=Z\). Hence \(E_r\) is Clifford with the lower shear as its symplectic action. Set \(T=H^*E_r^*H\). The identities \(H^*ZH=X\) and \(H^*XH=Z^{-1}\) show that \(T\) has the upper shear as its action and is a polynomial in \(X^r\). Because \(d\mid r^2\), \(X^r\) and \(Z^r\) commute, so \(T\) and \(E_r\) commute. Put \(A_0=Z^{-1}T\) and \(B_0=XE_r\). Using also that \(T\) commutes with \(X\) and \(E_r\) with \(Z\),
\[
B_0A_0=XZ^{-1}E_rT=\omega Z^{-1}TXE_r=\omega A_0B_0.
\]
Now choose a unit eigenvector $u$ of $B_0$, with eigenvalue $\lambda_1\in\mathbb T$. Then \(B_0A_0^ju=\lambda_1\omega^jA_0^ju\) for \(0\leq j<d\). These are $d$ distinct eigenvalues, so the corresponding vectors form an orthonormal basis. Therefore $B_0^d=\lambda_1^dI$, and similarly we can get $A_0^d=\lambda_2^dI$ for some \(\lambda_2\in\mathbb T\). We define \(A:=\frac{1}{\lambda_2}A_0\) and \(B:=\frac{1}{\lambda_1}B_0\), forming a valid conjugate pair. Then, Lemma~\ref{lem:conjugate_pairs}(ii) gives the corresponding unitary \(U\in\mathcal C_3\), as all products \(A^aB^b\) are Clifford. Finally, Lemma~\ref{lem:pauli-kernel-criteria}(iii) gives
\[
\Phi_U(a,b)=\rho([A])^a\rho([B])^b
=\begin{pmatrix}1&ra\\rb&1\end{pmatrix}.
\]
Thus \((a,b)\in K(U)\) if and only if \(ra=rb=0\) in \(\mathbb Z_d\), equivalently \((a,b)\in t\mathbb Z_d^2\). Since \(t>1\), this kernel contains no unimodular element, and Lemma~\ref{lem:pauli-kernel-criteria}(i) proves the claim.
\end{proof}

Lemmas~\ref{lem:squarefree-semiclifford} and ~\ref{lem:explicit-clifford-shears} give the following equivalence between square-free dimensions and semi-Clifford gates.

\begin{theorem}
\label{thm:semi-clifford-equivalent}
Let \(d\geq2\). Every gate in the single-qudit Clifford hierarchy is semi-Clifford if and only if \(d\) is square-free. Moreover, when \(d\) is not square-free, a non-semi-Clifford gate already exists at the third level.
\end{theorem}
\section{Third-level hierarchy gates are generalized semi-Clifford for arbitrary \texorpdfstring{$d$}{}}
\label{sec:gscg}
Section~\ref{ssec:not-sc-non-square-free} constructs third-level
gates that are not semi-Clifford in every non-square-free
dimension. Nevertheless, every third-level gate is generalized
semi-Clifford. We use the symplectic action from
Lemma~\ref{lem:symplectic-action-lagrangians} to find an invariant
Lagrangian in prime-power dimension. The Weyl-invariance
condition in the preliminaries then gives the required Pauli
algebra directly. Finally, CRT factorization reduces arbitrary
dimensions to the prime-power case.

\subsection{The prime-power case}
\label{ssec:prime-power-gsc}
Fix \(d=p^\alpha\), where \(p\) is prime and \(\alpha\geq1\).
We use the Lagrangians \(L_m=p^m\mathbb Z_{p^\alpha}e_1
\oplus p^{\alpha-m}\mathbb Z_{p^\alpha}e_2 \subseteq\mathbb Z_{p^\alpha}^2\) from Lemma~\ref{lem:lagrangian-classification}.

\begin{lemma}[Level-three generalized semi-Clifford theorem in prime-power dimension]
\label{lem:prime_power_gsc}
When \(d=p^\alpha\), every \(U\in\mathcal C_3\) is generalized semi-Clifford. Moreover, there exist \(0\leq r,s\leq\lfloor\alpha/2\rfloor\) and Clifford gates \(C_1,C_2\) such that \((C_2UC_1)\mathfrak A_{L_r}(C_2UC_1)^*=\mathfrak A_{L_s}\).
\end{lemma}
\begin{proof}
Let \(U\in\mathcal C_3\), put \(V_v:=UW_vU^*\), and let \(\mathcal S_U\) be its symplectic image, same as Lemma~\ref{lem:pauli-kernel-criteria}(iii), which is a finite \(p\)-group by the same lemma. Lemma~\ref{lem:symplectic-action-lagrangians} gives its action on \(\operatorname{Lag}(\mathbb Z_{p^\alpha}^2)\). Since \(\mathcal S_U\) is a finite \(p\)-group, the orbit--stabilizer theorem implies that every orbit
has \(p\)-power cardinality; see \cite[pp.~26--27]{HoltonSheehan1993}. On the other hand, Lemma~\ref{lem:lagrangian-classification}(iii) gives \(|\operatorname{Lag}(\mathbb Z_{p^\alpha}^2)|
=1+p+\cdots+p^\alpha\equiv1\pmod p\). The orbits partition this set, and every orbit with more than one element has cardinality divisible by \(p\). Hence there is a singleton orbit, corresponding to a Lagrangian fixed by every element of
\(\mathcal S_U\). Since each \(V_v\) is Clifford, Lemma~\ref{lem:clif_prop}(ii) gives
\(V_v\mathfrak A_LV_v^*=\mathfrak A_L\) for every \(v\).
Consequently, \(W_v(U^*\mathfrak A_LU)W_v^*=U^*(V_v\mathfrak A_LV_v^*)U =U^*\mathfrak A_LU\). This is precisely the defining condition for \(U^*\in\mathcal{GSC}\). By the adjoint invariance in Lemma~\ref{cor:lsc-clifford-invariant}, \(U\in\mathcal{GSC}\).

For the standard-form assertion, first by Lemma~\ref{lem:pauli-conjugation-equivalence}, there is a Lagrangian \(\widehat L\) with \(U^*\mathfrak A_LU=\mathfrak A_{\widehat L}\).  Lemma~\ref{lem:lagrangian-classification}(i) supplies symplectic maps \(S_1,S_2\) and indices \(r,s\) such that \(S_1(L_r)=\widehat L\) and \(S_2(L)=L_s\). By the surjectivity of \(\rho\), there exist \(C_1,C_2\in\mathcal{C}_2\) such that their symplectic representations are \(S_1, S_2\) respectively. Applying Corollary~\ref{lem:clif_prop}(i) to these two maps yields \((C_2UC_1)\mathfrak A_{L_r}(C_2UC_1)^*=C_2(U\mathfrak A_{\widehat L}U^*)C_2^*=\mathfrak A_{L_s}\)
\end{proof}

\begin{rmk}
\label{rmk:higher-level-difficult}
The symplectic action in Lemma~\ref{lem:symplectic-action-lagrangians} remains available at every hierarchy level, but the preceding proof does not automatically extend beyond level three. For \(U\in\mathcal C_{k+1}\) with \(k\geq3\), the conjugates \(UW_vU^*\in\mathcal C_k\) need not be Clifford. Although their projective classes still form an abelian group, the Clifford--symplectic homomorphism \(\rho\) is not defined on these elements in general. Thus, the proof no longer supplies a symplectic \(p\)-group action to which the fixed-point argument can be applied.
\end{rmk}

\subsection{Arbitrary dimension}
\label{ssec:arbitrary-dimension-gsc}
The Chinese-remainder factorization of the hierarchy and of Pauli algebras now extends the result to every \(d\geq2\).

\begin{theorem}[Level-three generalized semi-Clifford theorem]
\label{thm:level-three-gsc-arbitrary}
For every \(d\geq2\), \(\mathcal C_3\subseteq\mathcal{GSC}\).
\end{theorem}
\begin{proof}
Let \(U\in\mathcal C_3\). By Lemma~\ref{lem:coprime-hierarchy-factorization}, \(\Gamma U\Gamma^*=\bigotimes_j U_j\), where \(U_j\in\mathcal C_3(q_j)\). Lemma~\ref{lem:prime_power_gsc} gives Lagrangians \(L_j,\widehat L_j\subseteq\mathbb Z_{q_j}^2\) with \(U_j\mathfrak A_{\widehat L_j}U_j^*=\mathfrak A_{L_j}\).
Let \(L=\Theta^{-1}\Bigl(\prod_j L_j\Bigr)\) and \(\widehat L=\Theta^{-1}\Bigl(\prod_j\widehat L_j\Bigr)\). By Lemma~\ref{lem:crt-lag}, both are Lagrangian, and the corresponding Pauli algebras factor under \(\Gamma\). Therefore \(\Gamma(U\mathfrak A_{\widehat L}U^*)\Gamma^*
    =\bigotimes_j(U_j\mathfrak A_{\widehat L_j}U_j^*)
    =\bigotimes_j\mathfrak A_{L_j}
    =\Gamma\mathfrak A_L\Gamma^*.\)
Thus \(U\mathfrak A_{\widehat L}U^*=\mathfrak A_L\), proving \(U\in\mathcal{GSC}\).
\end{proof}

\bibliographystyle{amsalpha}
\bibliography{references.bib}

\appendix
\appendixpage
\section{\texorpdfstring{\((\mathcal{SC}_k)_{k\geq1}\), \((\mathcal N_k)_{k\geq1}\), \((\mathcal{LSC}_k)_{k\geq1}\), and \((\mathcal{GSC}_k)_{k\geq1}\) are strictly increasing}{}}
\label{app:strict_nest}

The recursive definition gives $\mathcal C_{k-1}\subseteq\mathcal C_k$ for $k\geq2$. We prove that $\mathcal C_k$ is strictly bigger than $\mathcal C_{k-1}$ by constructing a diagonal gate of at level $k$, for every $d\geq2$. Since a diagonal gate commutes with $Z$, the conjugate pair in Lemma~\ref{lem:conjugate_pairs} suggests focusing on its conjugation of $X$. The key observation is that this conjugation takes a discrete derivative of the phase function, as in the proof of Lemma~2 in \cite[Sec.~III]{Cui_2017}.
We reverse this operation to construct successive gates by discrete integration. More precisely, let $D_g|x\rangle=e^{2\pi i g(x)}|x\rangle$. To construct a diagonal gate $D_{g'}$ satisfying \(D_{g'}XD_{g'}^*=e^{2\pi i\alpha}XD_g\), we seek a phase function $f$ such that \(\Delta g'(x):=g'(x+1)-g'(x)
\equiv g(x)+\alpha\pmod{1}\), where the indices are understood modulo $d$.
A sufficient condition is the exact real-valued identity $\Delta g'(x)=g(x)+\alpha$, which we solve by cumulative summation \(g'(x)=g'(0)+\sum_{t=0}^{x-1}\bigl(g(t)+\alpha\bigr)\). For this solution to be consistent with cyclic indexing, we require $g'(d)=g'(0)$. Equivalently, the differences must sum to zero, that is, \(0=\sum_{x=0}^{d-1}\Delta g'(x) =\sum_{x=0}^{d-1}g(x)+d\alpha\). We therefore choose $\alpha=-\mu(g)$, where $\mu(g)=d^{-1}\sum_{x=0}^{d-1}g(x)$.
Subtracting the mean thus ensures cyclic consistency, while changing the target gate only by a global phase. This motivates the following recursive construction. For a real-valued function on \(\{0,\ldots,d-1\}\), put \(\mu(h)=d^{-1}\sum_{x=0}^{d-1}h(x)\). We define
\[
F_1(x):=x/d,\qquad
F_k(x):=\sum_{t=0}^{x-1}\bigl(F_{k-1}(t)-\mu(F_{k-1})\bigr)
\quad(k\geq2),
\qquad
D_k|x\rangle:=e^{2\pi iF_k(x)}|x\rangle.
\]

\begin{lemma}
\label{lem:conjugation}
We have \(D_1=Z\). For \(k\geq2\), set \(c_k=e^{-2\pi i\mu(F_{k-1})}\). Then \(D_kZD_k^*=Z\), \(D_kXD_k^*=X(c_kD_{k-1})\), and thus \(D_k\in\mathcal C_k\setminus\mathcal C_{k-1}\).
\end{lemma}
\begin{proof}
By construction,
\(F_k(x+1)-F_k(x)=F_{k-1}(x)-\mu(F_{k-1})\).
This includes the cyclic edge \(x=d-1\), because the right-hand side sums to zero. The conjugation identities follow, as does
\[
D_kX^aD_k^*=X^aE_{k,a},\qquad
E_{k,a}=c_k^a\prod_{r=0}^{a-1}X^{-r}D_{k-1}X^r
\quad(0\leq a<d).
\]
Starting with \(D_1=Z\), suppose inductively that \(D_{k-1}\in\mathcal D_{k-1}\). Since diagonal gates in the same level of Clifford hierarchy form a group, we have \(E_{k,a}\in\mathcal D_{k-1}\). Therefore, \(D_k(\lambda X^aZ^b)D_k^*=\lambda X^aE_{k,a}Z^b\in\mathcal C_{k-1}\), using Pauli multiplication invariance (the group property at level one and Lemma~\ref{lem:clif_prop} at higher levels). Hence \(D_k\in\mathcal D_k\).

To prove exactness, the cyclic adjacent ratio of the entries of \(D_2\) is \(e^{2\pi i(x/d-(d-1)/(2d))}\), which varies with \(x\). A diagonal Pauli has a constant adjacent ratio, so \(D_2\notin\mathcal C_1\). For \(k\geq3\), if \(D_k\in\mathcal C_{k-1}\), then \(X(c_kD_{k-1})\in\mathcal C_{k-2}\). Removing the scalar and Pauli factor would give \(D_{k-1}\in\mathcal C_{k-2}\), contradicting induction.
\end{proof}

Obviously, for fixed \(k\), \(D_k\) in Lemma~\ref{lem:conjugation} also belongs to \(\mathcal{SC}_k\setminus\mathcal{SC}_{k-1}\), \(\mathcal{N}_k\setminus\mathcal{N}_{k-1}\), \(\mathcal{LSC}_k\setminus\mathcal{LSC}_{k-1}\) and \(\mathcal{GSC}_k\setminus\mathcal{GSC}_{k-1}\). This shows that \((\mathcal{SC}_k)_{k\geq1}\), \((\mathcal N_k)_{k\geq1}\), \((\mathcal{LSC}_k)_{k\geq1}\), and \((\mathcal{GSC}_k)_{k\geq1}\) are strictly increasing. Moreover, the same examples provide strictness for any number \(n\geq1\) of qudits. M We define \(C_k^{(n,d)}\) to be the \(k\)-th level of Clifford hierarchy gates in \(n\) qudits of dimension \(d\). Indeed, \(U\in\mathcal C_k^{(1,d)}\) if and only if \(U\otimes I_d^{\otimes(n-1)}\in\mathcal C_k^{(n,d)}\). At level one, this follows from the uniqueness of Pauli tensor factors up to scalar phases. Inductively, test against tensor-product Pauli operators and apply the equivalence at level \(k-1\); the Pauli factors on the remaining sites can be removed without changing the level. Consequently, for every \(k\geq2\), \(D_k\otimes I_d^{\otimes(n-1)}\in\mathcal C_k^{(n,d)}\setminus\mathcal C_{k-1}^{(n,d)}\). 
\section{\texorpdfstring{Relations between \(\mathcal{SC}\), \(\mathcal N\), \(\mathcal{LSC}\), and \(\mathcal{GSC}\) in dimension nine}{}}
\label{app:dim9}
Throughout this appendix, \(d=9\), \(\eta=e^{2\pi i/27}\), \(\omega=\eta^3\), and \(\zeta=\eta^9\). We first record the diagonal Clifford characterization needed below. A diagonal unitary \(T\) is Clifford if and only if, up to a scalar, \(T\ket{x}=\omega^{ex^2+fx}\ket{x}\) for some \(e,f\in\mathbb Z_9\).
Indeed, if \(T\ket{x}=t_x\ket{x}\) is Clifford, then
\(TXT^*=\mu XZ^r\) for some \(\mu\in\mathbb T\) and \(r\in\mathbb Z_9\),
so \(t_{x+1}/t_x=\mu\omega^{rx}\). Multiplying these ratios around the nine-cycle gives \(\mu^9=1\). Their successive products yield a quadratic exponent, since \(2\) is invertible modulo nine. Conversely, quadratic phases conjugate \(X\) to a Pauli and fix \(Z\), hence are Clifford. Likewise, as presented in Lemma~\ref{lem:clif_prop}, a permutation gate is Clifford exactly when its permutation is affine with unit slope, as established in the preliminaries.

\begin{lemma}[Diagonal third-level gates in dimension nine]
\label{lem:diagonal-third-level-nine}
A diagonal unitary \(D\) belongs to \(\mathcal C_3\) if and only if
there are unique \(\lambda\in\mathbb T\) and \(a,b,c\in\{0,\ldots,8\}\) such that \(D\ket{x}=\lambda\eta^{ax^3+3bx^2+3cx}\ket{x}\) for every \(x\in\mathbb Z_9\).
\end{lemma}

\begin{proof}
For integers \(a,b,c\), the polynomial \(h(x)=ax^3+3bx^2+3cx\) satisfies
\(h(x+9)\equiv h(x)\pmod{27}\), so the formula is well-defined. Its first difference is \(h(x+1)-h(x)=3ax^2+3(a+2b)x+a+3b+3c\). Thus \(DZD^*=Z\) and \(DXD^*=\eta^{a+3b+3c}XT\), where \(T\ket{x}=\omega^{ax^2+(a+2b)x}\ket{x}\) is Clifford. Since the Pauli group is generated by \(X,Z\) and scalar phases, and \(\mathcal C_2\) is a group, every conjugated Pauli is Clifford. Hence \(D\in\mathcal C_3\).

Conversely, if \(D\in\mathcal C_3\), then \(E=X^*DXD^*\) is a diagonal Clifford, so \(E\ket{x}:=\mu\omega^{ex^2+fx}\ket{x}\). Choose \(a,b\in\{0,\ldots,8\}\) with \(a\equiv e\) and \(a+2b\equiv f\pmod9\), and set
\(F\ket{x}:=\eta^{ax^3+3bx^2}\ket{x}\). The first-difference formula shows that the adjacent diagonal-entry ratios of \(DF^*\) are constant. Their product is one, so that constant is \(\omega^c\) for a unique \(c\in\{0,\ldots,8\}\). Therefore \(DF^*=\lambda Z^c\), giving the required form.

For uniqueness, \(\lambda=\langle0|D|0\rangle\). The normalized adjacent ratios determine \(e,f\) uniquely: their values at \(x=1,2\) determine the coefficients of \(ex^2+fx\), since \(2\) is a unit in \(\mathbb Z_9\). Thus \(a,b\), and then \(c\), are uniquely determined.
\end{proof}

\subsection{Relations at level three}
\label{ssec:sc_3}
We define \(\Pi_9\ket{x}=\ket{p(x)}\), where \(p(x)=-3x^2+x+3\pmod9\). It cycles \(0,3,6\) and fixes the other six labels, with inverse
\(p^{-1}(x)=3x^2+x-3\pmod9\). Direct substitution gives \(\Pi_9X^aZ^b\Pi_9^*\ket{x} =\omega^{3bx^2+bx-3b}\ket{(1+3a)x+a-3a^2}\).
The permutation is affine with unit slope and the phase is quadratic, so every conjugate is Clifford. Since \(\Pi_9\) is a permutation gate, \(\Pi_9\in\mathcal N_3\). For a conjugate to be Pauli, its permutation must be a translation, which forces \(3\mid a\). The second finite difference of its phase exponent is \(6b\pmod9\), so \(3\mid b\) is also necessary.
Conversely, these divisibilities reduce the conjugate to \(X^aZ^b\).
Consequently, \(K(\Pi_9)=3\mathbb Z_9^2\), and thus \(\Pi_9\in(\mathcal N_3\cap\mathcal{LSC}_3)\setminus\mathcal{SC}_3.\) The membership assertions follow from Lemma~\ref{lem:pauli-kernel-criteria}, since \(3\mathbb Z_9^2\) is Lagrangian and contains no unimodular element.

\begin{proof}[Proof of Proposition~\ref{prop:c3-lsc-gsc-d9}]
Fix \(U\in\mathcal C_3\) and put \(\mathcal S=\mathcal S_U=\langle S_a,S_b\rangle\), where \(S_a=\Phi_U(1,0)\) and \(S_b=\Phi_U(0,1)\).
By Lemma~\ref{lem:pauli-kernel-criteria}, \(\mathcal S\) is abelian, its order divides \(81\), and \(|K(U)|\,|\mathcal S|=81\). It therefore suffices to prove \(|\mathcal S|\leq9\). Let \(\mu:\mathcal S\to\overline{\mathcal S}\) denote reduction
modulo three.

If \(\overline{\mathcal S}=\{I\}\), both generators have the form
\(I+3M\) and their cubes are the identity. Since they commute,
\(|\mathcal S|\leq9\).

Otherwise, \(|\overline{\mathcal S}|\) divides both \(81\) and \(|\operatorname{SL}_2(\mathbb F_3)|=24\), so \(|\overline{\mathcal S}|=3\).
Choose \(S\in\mathcal S\) with \(\overline S\neq I\). This nonidentity matrix of order three is unipotent. After a change of basis in \(\operatorname{GL}_2(\mathbb F_3)\), lifted to \(\operatorname{GL}_2(\mathbb Z_9)\), we may assume \(\overline S=\left(\begin{smallmatrix}1&1\\0&1\end{smallmatrix}\right)\). Every \(T\in\ker\mu\) has a unique expression \(T=I+3M\) with \(M\in\operatorname{M}_2(\mathbb F_3)\), where \(3M\) is interpreted using any lift to \(\operatorname{M}_2(\mathbb Z_9)\). Since \(T\) commutes with \(S\), we have
\(M\overline S=\overline SM\), hence \(M=\left(\begin{smallmatrix}u&v\\0&u\end{smallmatrix}\right)\).
The identity \(\det T=1\) gives \(\operatorname{tr}M=2u=0\) in \(\mathbb F_3\), so \(u=0\).
Thus \(|\ker\mu|\leq3\), and \(|\mathcal S|=|\overline{\mathcal S}|\,|\ker\mu|\leq9\).

In either case, \(|K(U)|\geq9\), so Lemma~\ref{lem:pauli-kernel-criteria}(iii) gives
\(U\in\mathcal{LSC}_3\). Together with \(\mathcal{LSC}_3\subseteq\mathcal{GSC}_3
\subseteq\mathcal C_3\), this proves (i). Part (ii) follows immediately from \(\mathcal N_3\subseteq\mathcal C_3\).
\end{proof}

\subsection{Separations at higher levels}
\label{app:dim9-separations}
We construct \(Q_9\in\mathcal N_4\setminus\mathcal{LSC}\) and \(V_9\in\mathcal{GSC}_4\setminus(\mathcal N\cup\mathcal{LSC})\). Since the excluded classes have no hierarchy restriction, these examples give the corresponding separations for every \(k\geq4\).

\subsubsection{A permutation--diagonal gate outside \(\mathcal{LSC}\)}
Set \(T_9\ket{x}=\eta^{x^4+x^2}\ket{x}\). This is well-defined modulo nine: replacing \(x\) by \(x+9\) changes the exponent by \(9(x^3+2x)\pmod{27}\), which vanishes. We define diagonal gates \(E_a\ket{x}=\eta^{h_a(x)}\ket{x}\), where \(h_a(x)=4ax^3+6a^2x^2+(4a^3+2a)x+a^4+a^2\). Then \(T_9X^aT_9^*=X^aE_a\). The quadratic and linear coefficients of \(h_a\) are divisible by three, since \(a^3\equiv a\pmod3\). Lemma~\ref{lem:diagonal-third-level-nine} therefore gives \(E_a\in\mathcal D_3\). Here the cubic coefficient can be reduced modulo nine by adjusting the linear coefficient, using \(9x^3\equiv9x\pmod{27}\). Thus \(T_9\in\mathcal D_4\).

Put \(Q_9=\Pi_9T_9\). For \(y=p^{-1}(x)\), we have \(y\equiv x\pmod3\), \(y^2\equiv x^2\pmod9\), and \(y^3\equiv x^3\pmod{27}\). Since \(y-x\equiv3(x^2-1)\pmod9\), we have \(h_a(p^{-1}(x))-h_a(x) \equiv3(4a^3+2a)(x^2-1)\pmod{27}\). The added term preserves the form in Lemma~\ref{lem:diagonal-third-level-nine}, so
\(\Pi_9E_a\Pi_9^*\in\mathcal D_3\). Consequently,
\[
Q_9X^aZ^bQ_9^*
=(\Pi_9X^aZ^b\Pi_9^*)(\Pi_9E_a\Pi_9^*)\in\mathcal C_3,
\]
by Clifford invariance. Therefore \(Q_9\in\mathcal N_4\).

For this conjugate to be Pauli, its permutation must be a translation, forcing \(a=3a_0\) with \(a_0\in\{0,1,2\}\). Its action then becomes \(Q_9X^{3a_0}Z^bQ_9^*\ket{x} =\omega^{g(p^{-1}(x))}\ket{x+3a_0}\), where \(g(x)=4a_0x^3+(2a_0+b)x+3a_0^2\pmod9\). The second finite difference of \(g\circ p^{-1}\) is \(6a_0x+6b\pmod9\). If the phase is affine, its values at \(x=0,1\) force \(3\mid b\) and \(a_0=0\). These conditions also suffice.
Hence \(K(Q_9)=\{(0,0),(0,3),(0,6)\}\), which contains no Lagrangian because every Lagrangian in \(\mathbb Z_9^2\) has order nine. Thus \(Q_9\notin\mathcal{LSC}\).

\subsubsection{A generalized semi-Clifford gate outside both intermediate classes}
Let \(L_{\mathrm{nf}}=3\mathbb Z_9^2\). We use the joint eigenbasis
\[
\ket{y,z}_{\mathrm{ft}}
=\frac1{\sqrt3}\sum_{t=0}^2\zeta^{-zt}\ket{y+3t},
\qquad y,z\in\{0,1,2\}.
\]
It satisfies \(X^3\ket{y,z}_{\mathrm{ft}}=\zeta^z\ket{y,z}_{\mathrm{ft}}\) and \(Z^3\ket{y,z}_{\mathrm{ft}}=\zeta^y\ket{y,z}_{\mathrm{ft}}\). The nine joint eigenvalue pairs are distinct, so \(\mathfrak A_{L_{\mathrm{nf}}}\) is precisely the algebra diagonal in this basis. We define
\[
V_9\ket{y,z}_{\mathrm{ft}}
:=\eta^{y^3}\zeta^{yz^2}\ket{y,z-y^2}_{\mathrm{ft}}.
\]
The output labels are taken modulo three, whereas \(y^3\) in the exponent of \(\eta\) uses \(y\in\{0,1,2\}\). This operator permutes the basis vectors up to phases, so it is unitary and normalizes \(\mathfrak A_{L_{\mathrm{nf}}}\). Therefore \(V_9\in\mathcal{GSC}\).

We first prove \(V_9\in\mathcal C_4\). For a polynomial function \(f:\mathbb F_3^2\to\mathbb F_3\), define
\(D_f\ket{y,z}_{\mathrm{ft}}=\zeta^{f(y,z)}\ket{y,z}_{\mathrm{ft}}\). If \(f\) is affine, then \(D_f\) is a Pauli, since the phases \(\zeta^y\) and \(\zeta^z\) represent \(Z^3\) and \(X^3\).
A Pauli \(W=X^aZ^b\) translates the labels by \((y,z)\mapsto(y+\bar a,z-\bar b)\), up to phases, where bars denote reduction modulo three. Consequently, \(D_fWD_f^*=WD_{f(y+\bar a,z-\bar b)-f(y,z)}\). Finite differences lower the total degree. Thus, a polynomial of degree at most two gives a Clifford \(D_f\), and one of degree at most three gives \(D_f\in\mathcal C_3\).

Let \(P\ket{x}=\eta^{x^3}\ket{x}\). Lemma~\ref{lem:diagonal-third-level-nine} gives \(P\in\mathcal D_3\),
and expansion of \((y+3t)^3\) modulo 27 gives \(P\ket{y,z}_{\mathrm{ft}}=\eta^{y^3}\ket{y,z-y^2}_{\mathrm{ft}}\).
Hence \(V_9=PD_{yz^2}\). For \(W=X^aZ^b\), write \(D_{yz^2}WD_{yz^2}^*=WD_f\), where \(f(y,z)=(y+\bar a)(z-\bar b)^2-yz^2\). Conjugation by \(P\) replaces \(z\) by \(z+y^2\). Using \(y^3=y\) on \(\mathbb F_3\), we obtain \(f(y,z+y^2)=-\bar a y^2z+q(y,z)\), where \(q\) has degree at most two. Therefore, \(V_9WV_9^*=(PWP^*)D_qD_{-\bar a y^2z}\in\mathcal C_3\), as the first two factors are Clifford and the last lies in \(\mathcal C_3\). This proves \(V_9\in\mathcal C_4\).

Next, suppose \(V_9\in\mathcal N\). By Lemma~\ref{lem:unimodular-direction-criteria}(ii), there are unimodular labels \((a,b)\) and \((c,e)\) such that \(G=V_9X^aZ^bV_9^*\) commutes with \(W=X^cZ^e\). Up to a scalar independent of \(y,z\), we have \(G^3\ket{y,z}_{\mathrm{ft}} =\zeta^{a(z+y^2)+by}\ket{y,z}_{\mathrm{ft}}\). Since \(W\) translates the labels by \((c,-e)\) modulo three, commutation with \(G^3\) requires \(a(2cy+c^2-e)+bc=0\) in \(\mathbb F_3\) for every \(y\). If \(a\not\equiv0\pmod3\), the linear and constant coefficients give \(c\equiv e\equiv0\pmod3\), contradicting admissibility. Thus \(a\equiv0\) and \(b\not\equiv0\); the same identity then forces \(c\equiv0\), so \(e\not\equiv0\).

Write \(a=3a_0\) and \(c=3c_0\), with \(a_0,c_0\in\{0,1,2\}\). The defining action of \(V_9\) gives
\[
\begin{aligned}
G\ket{y,z}_{\mathrm{ft}}
&=\lambda_y\zeta^{(a_0+by)z}\ket{y,z-b}_{\mathrm{ft}},\\
W\ket{y,z}_{\mathrm{ft}}
&=\mu_y\zeta^{c_0z}\ket{y,z-e}_{\mathrm{ft}},
\end{aligned}
\]
where \(\lambda_y,\mu_y\) are phases independent of \(z\). Commutation requires \(e(a_0+by)=c_0b\) in \(\mathbb F_3\) for every \(y\). Subtracting the equations for \(y=1,0\) gives \(eb=0\), contradicting \(b,e\not\equiv0\pmod3\). Therefore \(V_9\notin\mathcal N\).

Finally, \(\mathcal{SC}\subseteq\mathcal N\) and Lemma~\ref{lem:unimodular-direction-criteria}(i) imply that
\(K(V_9)\) contains no unimodular element. Thus \(K(V_9)\subseteq3\mathbb Z_9^2\). For \(t,u\in\mathbb Z_3\), direct conjugation gives
\[
V_9X^{3t}Z^{3u}V_9^*\ket{y,z}_{\mathrm{ft}}
=\zeta^{t(z+y^2)+uy}\ket{y,z}_{\mathrm{ft}}.
\]
A Pauli diagonal in this basis commutes with \(X^3,Z^3\), so its label lies in \(L_{\mathrm{nf}}^\perp=L_{\mathrm{nf}}\). Its phase exponent is therefore affine in \(y,z\).
The displayed exponent is affine exactly when \(t=0\), giving \(K(V_9)=\{(0,0),(0,3),(0,6)\}\).
This kernel has order three and contains no Lagrangian, so \(V_9\notin\mathcal{LSC}\).
Both \(Q_9\) and \(V_9\) have exact level four, because Proposition~\ref{prop:c3-lsc-gsc-d9} places every third-level gate in \(\mathcal{LSC}\).

\end{document}